\documentclass[lettersize,journal]{IEEEtran}
\usepackage[colorlinks,urlcolor=blue,linkcolor=blue,citecolor=blue]{hyperref}
\usepackage{color,array}

\usepackage{graphicx}
\usepackage{amsmath,amssymb,amsfonts,amsthm}

\newtheorem{assumption}{Assumption}
\newtheorem{theorem}{Theorem}
\newtheorem{lemma}{Lemma}
\newtheorem{remark}{Remark}
\newtheorem{definition}{Definition}

\usepackage{float}
\usepackage{algorithm}
\usepackage{algorithmic}
\usepackage{subcaption}
\usepackage{multirow}
\usepackage{booktabs}

\begin{document}

% Transferred from your template text
\title{Direct Data Driven Natural Gradient Descent for Control} 

% Reconfigured author block into the standard IEEEtran journal structure
\author{Ramin~Esmzad,
        Farnaz~Adib~Yaghmaie,
        Bahare~Kiumarsi,
        and~Hamidreza~Modares,~\IEEEmembership{Member,~IEEE}%
% \thanks{This work was supported by the National Science Foundation under award ECCS-2227311. Farnaz Adib Yaghmaie is supported by the Excellence Center at Linköping–Lund in Information Technology (ELLIIT), ZENITH, and partially by Sensor informatics and Decision-making for the Digital Transformation (SEDDIT). This work was partly performed within the Competence Center SEDDIT-Sensor Informatics and Decision Making for the Digital Transformation, supported by Sweden’s Innovation Agency within the research and innovation program Advanced Digitalization.}%
\thanks{Ramin Esmzad and Hamidreza Modares are with the Department of Mechanical Engineering, Michigan State University, East Lansing, MI 48824 USA (e-mail: modaresh@msu.edu).}%
\thanks{Farnaz Adib Yaghmaie is with the Department of Electrical Engineering, Linköping University, 58183 Linköping, Sweden.}%
\thanks{Bahare Kiumarsi is with the Department of Electrical and Computer Engineering, Michigan State University, East Lansing, MI 48824 USA.}%
}

% Optional: Page headers for IEEE Transactions
\markboth{IEEE TRANSACTIONS ON CONTROL SYSTEMS TECHNOLOGY,~Vol.~XX, No.~X, Month~202X}{Esmzad \MakeLowercase{\textit{et al.}}: Direct Data Driven Natural Gradient Descent for Control}

\maketitle

\begin{abstract}
This paper introduces a novel direct data-driven control framework based on Natural Gradient Descent (NGD) to design interpretable and robust closed-loop policies without requiring explicit model identification. We propose two data-driven NGD formulations that incorporate the closed-loop covariance matrix through the Fisher Information Matrix (FIM), allowing gradient updates to be preconditioned according to the system's intrinsic uncertainty. Leveraging two distinct data-based parameterizations of the closed-loop system, our method enables stability-guaranteed policy synthesis directly from data. We provide theoretical guarantees for contraction and convergence using semidefinite programs (SDPs) and validate our framework in both simulations and on hardware on a ROSbot XL platform. The results demonstrate intuitive features compared to linear-quadratic regulator (LQR) and standard data-driven baselines, particularly in terms of convergence speed, robustness, and control interpretability. This work bridges the gap between trajectory-oriented natural gradient methods and practical data-driven control design.
\end{abstract}

\begin{IEEEkeywords}
Data-Driven Control, 
 Linear Matrix Inequalities, Natural Gradient Descent, Robotics. 
% Enter key words or phrases in alphabetical order, separated by commas. For a list of\goodbreak suggested keywords, send a blank e-mail to \href{mailto:keywords@ieee.org}{keywords@ieee.org} or visit\goodbreak \href{http://www.ieee.org/organizations/pubs/ani_prod/keywrd98.txt}{http://www.ieee.org/organizations/pubs/ani\_prod/keywrd98.txt}
\end{IEEEkeywords}

\maketitle

\section{INTRODUCTION}
The integration of machine learning into control theory is revolutionizing how {dynamical} systems are controlled to achieve desired specifications. Machine learning offers the ability to learn from data, enabling control systems to adapt to changing conditions and optimize performance. Gradient descent (GD) is a fundamental algorithm for optimizing cost functions from data. In learning-based control, GD tunes the parameters of a system model or control policy to minimize modeling errors or implementation costs \cite{ruder2017overviewgd,laborde2020lyapunov,martens2020new}. However, traditional GD-based control designs predominantly rely on indirect trajectory shaping via optimizing predefined cost functions. However, these methods require tedious, iterative tuning of cost parameters (e.g., weighting matrices in the linear-quadratic regulator (LQR)) and risk reward hacking, where resulting trajectories misalign with intended behavior \cite{skalse2022defining}. This drives the need for frameworks that directly shape closed-loop trajectories with minimal heuristic adjustments. Alternatively, control tools can analyze GD optimization algorithms by treating them as dynamical systems to enable robust closed-loop optimization \cite{HAUSWIRTH2024100941,lessard2016analysis,padmanabhan2024analysis,nayyer2022passivity}. This paper does not focus on this direction of using control to improve optimization algorithms. To address these challenges, \cite{esmzad2024gdl,esmzad2026natural} embeds gradient descent dynamics directly into the closed loop by using GD to parameterize the system dynamics. This forces system states to naturally evolve along a Lyapunov-based GD algorithm, providing geometric clarity, enhanced interpretability, and explicit robustness through systematic trajectory shaping.

However, the framework in \cite{esmzad2024gdl,esmzad2026natural} requires known system dynamics, leaving it unsuited for data-driven scenarios with unknown parameters. While direct data-driven control strategies bypass model identification to design robust controllers directly from data \cite{DEPERSIS2021109548,dorfler2022bridging,ESMZAD2025112197}, they have yet to leverage the geometric advantages of Natural Gradient Descent (NGD) \cite{6790500}. NGD accelerates convergence by accounting for the intrinsic geometry of the parameter space. In control applications, this geometry explicitly encodes robustness by capturing state-space uncertainty through the Fisher Information Matrix (FIM)~\cite{esmzad2026natural}.

To bridge this gap, this paper proposes a direct data-driven NGD control framework that learns closed-loop policies directly from data, bypassing explicit system identification. Our main contributions are summarized as follows:
\begin{itemize}
    \item \textbf{Direct Data-Driven NGD Control:} By leveraging dual data-based system parameterizations~\cite{zhao2024data}, our framework embeds NGD geometry into the system dynamics via the Fisher Information Matrix (FIM). Rather than optimizing a control cost or tracking a pre-planned trajectory, this forces the closed-loop states to follow an intrinsic natural-gradient flow, yielding a structured, robust, and highly interpretable family of trajectories that fundamentally distinguishes it from classical LQR and contemporary data-driven policy optimization.
    \item \textbf{Uncertainty-Informed Feedback:} The framework preconditions gradient updates using the closed-loop covariance (FIM). This encodes system uncertainty and disturbances directly into the feedback law, enabling explicit, interpretable robustness without requiring heuristic cost adjustments.
    \item \textbf{Rigorous Theoretical Guarantees:} We establish exact stability and convergence conditions tailored to NGD-based closed-loop dynamics. Furthermore, we provide supporting complexity analysis, including a high-probability sample bound for data richness (Lemma~3) and a geometric iteration-complexity bound (Lemma~4).
\end{itemize}

Finally, extensive validation of the developed algorithms is performed on a real robotic platform (ROSbot XL) for point regulation tasks, demonstrating significant improvements over traditional data-driven and model-based methods in terms of interpretability of the approach. \vspace{6pt}

{
\noindent \textit{Notations:} 
Throughout this paper, $\mathbb{R}^{n}$ and $\mathbb{R}^{n \times m}$ denote the sets of $n$-dimensional real vectors and $n \times m$ real matrices, respectively. {$\mathbb S^{d-1}$ denotes the unit sphere in $\mathbb R^{d}$.} For a matrix $A$, $\mathrm{Tr}$ denotes its trace, $A^\top$ denotes its transpose, and $\|A\|$ represents its induced spectral norm (2-norm). For a square matrix $A \in \mathbb{R}^{n \times n}$, $\rho(A)$ denote the spectral radius; $\rho(A)=\max \{\vert \lambda_1 \vert,...,\vert \lambda_n\vert\}$ where $\lambda_i,:i=1,...,n$ are the eigenvalues of $A$. For a symmetric matrix $A$, $A \succ 0$ ($A \succeq 0$) signifies that $A$ is strictly positive definite (positive semi-definite). The maximum and minimum eigenvalues of a symmetric matrix $A$ are denoted by $\lambda_{\max}(A)$ and $\lambda_{\min}(A)$, respectively, and its condition number is defined as $\kappa(A) = \lambda_{\max}(A)/\lambda_{\min}(A)$. The $n \times n$ identity matrix is denoted by $I_n$. The standard gradient of a differentiable scalar function $f: \mathbb{R}^n \to \mathbb{R}$ with respect to ${x} = [x_1, \dots, x_n]^\top$ is defined as 
   $ \nabla f(\mathbf{x}) = \left[ \frac{\partial f}{\partial x_1}, \frac{\partial f}{\partial x_2}, \dots, \frac{\partial f}{\partial x_n} \right]^\top.$ The notation
\(
\{ z_k \}_{k=0}^{N}
\)
denotes the finite sequence of elements indexed by $k$
from $0$ to $N$, i.e.,
\(
\{ z_k \}_{k=0}^{N}
=
\{ z_0, z_1, z_2, \dots, z_N \}
\). If $k$ denotes discrete time, this notation describes a trajectory over the time horizon $0 \le k \le N$. Finally, the notation $\mathcal{N}(\mu_k, \Sigma_k)$ denotes a multivariate Gaussian distribution at iteration $k$, characterized by the mean vector $\mu_k \in \mathbb{R}^n$ and the covariance matrix $\Sigma_k \succeq 0$, and $\mathbb{E}[\cdot]$ and $\Pr(\cdot)$ denote the expected value and probability of random variables, respectively, evaluated with respect to the underlying probability measure. 

\section{System Model and Problem Description}\label{sec:preliminaries}
This section outlines the stochastic linear system under consideration and the assumptions used throughout the paper. We describe the closed-loop formulation, define the state distribution dynamics, and provide the two common direct data-based formulations for linear time-invariant (LTI) systems. 

\subsection{System Model}
\noindent Consider a stochastic LTI discrete-time system of the form 
\begin{align} 
x_{k+1}=A \,x_k+B \, u_k+{\omega_k},
\label{eq:syst}
\end{align} 
where $k \in \mathbb{N}$, $x_k \in \mathbb{R}^n$ is the system's state, and $u_k \in \mathbb{R}^m$ is the control input. $A$ and $B$ are transition and input matrices of appropriate dimensions, respectively. The noise of the system $\omega_k \in \mathbb{R}^n$ is i.i.d. with a Gaussian distribution $\mathcal{N}(0, W)$ where $W \in \mathbb{R}^{n\times n} \succ 0$ is the covariance {matrix}. Throughout this paper, a linear control policy is used as
\begin{align} 
u_{k}=K \,x_k, \quad K \in \mathbb{R}^{m \times n}.
\label{eq:control}
\end{align} 
% The dynamical system \eqref{eq:syst} evolves stochastically due to the additive Gaussian noise \( \omega_k \)~\cite{kalman1960new}.
Using this control policy, the distribution of the state $x_k$ at time $k$ is given by \( x_k \sim \mathcal{N}(\mu_k, \Sigma_k) \),
where the mean $\mu_k$ and the covariance $\Sigma_k$ are updated respectively as follows \cite{kalman1960new}
\begin{align}
\mu_{k+1} &= (A + BK)\mu_k, \\
\Sigma_{k+1} &= (A + BK) \Sigma_k (A + BK)^\top + W. \label{eq:cov}
\end{align}
The term \( (A + BK)\Sigma_k(A + BK)^\top \) in \eqref{eq:cov} captures how closed-loop dynamics influence the propagation of uncertainty. The following assumptions are in order.

\begin{assumption} \label{assu:AB}
    The pair $(A, B)$ is unknown but controllable. 
\end{assumption}
\begin{assumption} \label{assu:W}
    $W$ is unknown but can be estimated from the collected data~\cite{ESMZAD2025112197} or can be used as a tuning parameter of the controller in the case that we have no access to it.
\end{assumption}
\begin{definition}\label{Def:InvSet} \textit{$\lambda-$Contractive and Positive Invariant Sets: }
Consider the system \eqref{eq:syst}. Let $\lambda \in (0,1]$. The set $\mathcal{P}$ is considered $\lambda$-contractive if for any $x_k \in \mathcal{P}$ it holds that $x_{k+1} \in \lambda \mathcal{P}$. When $\lambda = 1$, $\mathcal{P}$ is positive invariant~\cite{blanchini2008set}.
\end{definition}

\subsection{Uncertainty-Aware Data-based System Description}
This subsection provides two data-based representations of the LTI system (1), inspired by \cite{DEPERSIS2021109548} and \cite{hotz1987covariance}. The first one, called direct parameterization, uses the collected input-state data to find a closed-loop system parameterization. The second one, called indirect parameterization or covariance parameterization, parametrizes the closed-loop system using the sample covariance of the collected data.

Suppose that we have collected $N$ sequences of input-state data as
\begin{subequations}
\label{eq:collected_data}
\begin{align} 
    U_0 &= \begin{bmatrix}
        u_0 & u_1 & \cdots & u_{N-1}
    \end{bmatrix}\in \mathbb{R}^{m \times N}, \\
    X_0 &= \begin{bmatrix}
        x_0 & x_1 & \cdots & x_{N-1}
    \end{bmatrix}\in \mathbb{R}^{n \times N}, \\
    X_1 &= \begin{bmatrix}
        x_1 & x_2 & \cdots & x_{N}
    \end{bmatrix}\in \mathbb{R}^{n \times N}.
\end{align}
\end{subequations}
%where $U_0 \in \mathbb{R}^{m \times N}$, $X_0\in \mathbb{R}^{n \times N}$, and $X_1 \in \mathbb{R}^{n \times N}$. 
The corresponding noise sequences are 
\begin{align}
    \Omega_0 = \begin{bmatrix}
        \omega_0 & \omega_1 & \cdots & \omega_{N-1}
    \end{bmatrix}\in \mathbb{R}^{n \times N},
\end{align}
in which we have no access to, and we treat its columns as random variables.
Based on the system dynamics (1), the collected data satisfy~\cite{DEPERSIS2021109548}
\begin{align}
    X_1 = A X_0 + B U_0 + \Omega_0. \label{eq:data-integrity}
\end{align}
\begin{assumption} \label{assum:datarank}
The collected data matrix $D_0 = \begin{bmatrix}
    U_0^T & X_0^T
\end{bmatrix}^T$ is sufficiently rich, i.e., it has full row rank. That is, $ rank(D_0) = m+n.$
\end{assumption}

\subsubsection{Direct Parameterization}
\label{subsub:dir}
By multiplying both sides of \eqref{eq:data-integrity} on the right by $G \in \mathbb{R}^{N \times n}$~\cite{DEPERSIS2021109548}, one has 
\begin{align}
    X_1 G = A X_0 G + B U_0 G + \Omega_0 G. \label{eq:x1g}
\end{align}
If $G$ satisfies 
\begin{align}
    \begin{bmatrix}
        K \\ I
    \end{bmatrix} = D_0 G. \label{eq:KD0G}
\end{align}
Then, the closed-loop matrix parameterization using the collected input-state data is
\begin{align}
    A+BK &= (X_1 - \Omega_0) G. \label{eq:ABK1}
\end{align}
Finally, the data-based parameterization of the closed-loop system can be represented as 
\begin{align}
    x_{k+1} = (X_1 - \Omega_0) G x_k + \omega_k. \label{eq:closed_db1} 
\end{align}
Since we have no access to the noise data matrix $\Omega_0$, we can treat its columns as a random variable with the same probabilistic characteristics as $\omega_k$.

\begin{lemma}[Closed-loop covariance for the direct parameterization]
Consider the data-based closed-loop representation
\eqref{eq:closed_db1}. 
The distribution of the state $x_k$ at time $k$ for this representation is given by \( x_k \sim \mathcal{N}(\mu_k, \Sigma_k) \) with the update rules as
\begin{align}
\mu_{k+1} &= X_1 G \mu_k, \label{eq:db_mean1} \\
\Sigma_{k+1} &= X_1 G \Sigma_k G^\top X_1^\top + \mathrm{Tr}(G \Sigma_k G^\top)W + W, \label{eq:db_cov1}
\end{align}
where the additional term $\mathrm{Tr}(G \Sigma_k G^\top) W$ arises from
propagating the uncertainty in the unknown noise samples $\Omega_0$ through
the data-based closed-loop map.
\end{lemma}

\begin{proof}
By taking the expectation from both sides of \eqref{eq:closed_db1}, the mean dynamics in \eqref{eq:db_mean1} is obtained by considering the independence of $\Omega_0$, $G$, and $x_k.$ 
Expanding the second moment yields
\begin{align}
&\mathbb{E} [x_{k+1}  x_{k+1}^T] = \nonumber \\
&\mathbb{E} [((X_1 - \Omega_0) G x_k+\omega_{k})  ((X_1 - \Omega_0) G x_k+\omega_{k})^\top], \nonumber \\ 
&\Sigma_{k+1} = X_1 G\mathbb{E} [x_k x_k^\top] G^\top X_1^\top + \mathbb{E} [\Omega_0 G x_k x_k^\top G^\top \Omega_0^\top] \nonumber \\
& +\mathbb{E} [\omega_k \omega_k^\top] \nonumber  \\
 &= X_1 G \Sigma_k G^\top X_1^\top + \mathbb{E} [\Omega_0 G \Sigma_k G^\top  \Omega_0^\top] +W, \label{eq:covariancemid1}
\end{align}
and by substituting $\mathbb{E} [\Omega_0 G \Sigma_k G^\top  \Omega_0^\top]  = \mathrm{Tr}(G \Sigma_k G^\top)W $ in \eqref{eq:covariancemid1}, the uncertainty-aware closed-loop data-driven covariance equation \eqref{eq:db_cov1} results. This completes the proof. 
\end{proof}

\begin{remark} \label{rem:cov1}
   Compared with \cite{DEPERSIS2021109548}, \eqref{eq:db_cov1} contains the additional term \(\mathrm{Tr}(G \Sigma_k G^\top)W\), which accounts for the uncertainty propagated by the noise term \(\Omega_0 G\) in \eqref{eq:ABK1}.
\end{remark}

\subsubsection{Indirect (Covariance) Parameterization}
{Unlike the direct parameterization in Sec. II.B.\ref{subsub:dir}, here we leverage the sample covariance of input-state data rather than through state–input snapshots alone.} This approach is widely used in the data-driven Riccati equation~\cite{rantzer2024data,rantzer2024lqdc} and other contexts, such as the covariance control theory~\cite{hotz1987covariance} to handle streaming data efficiently~\cite{zhao2024data}. The sample covariance is defined as 
\begin{align}
    \Phi = \frac{1}{N}D_0 D_0^\top = \begin{bmatrix}
        \frac{U_0 D_0^\top}{N} \\ \frac{X_0 D_0^\top}{N}
    \end{bmatrix} = \begin{bmatrix}
        \bar{U}_0 \\ \bar{X}_0
    \end{bmatrix}. \label{eq:sampCov}
\end{align}
If Assumption \ref{assum:datarank} holds, then $\Phi$ is positive definite. If we multiply both sides of \eqref{eq:data-integrity} on the right by $\frac{D_0^\top H}{N}$, where $H \in \mathbb{R}^{(n+m) \times n} $, then 
\begin{align}
    X_1 \frac{D_0^\top H}{N} = A X_0 \frac{D_0^\top H}{N} + B U_0 \frac{D_0^\top H}{N} + \Omega_0 \frac{D_0^\top H}{N}. \label{eq:x1g}
\end{align}
If $H$ satisfies
\begin{align}
     \begin{bmatrix}
        K \\ I
    \end{bmatrix} = \Phi H, \label{eq:KD0H}
\end{align}
and considering $\bar{\Omega}_0=\frac{\Omega_0 D_0^\top}{N}$ and $\bar{X}_1=\frac{X_1 D_0^\top}{N}$, then the closed-loop matrix parameterization will be 
\begin{align}
    A+BK &= (\bar{X}_1 - \bar{\Omega}_0) H. \label{eq:ABK2}
\end{align}
Finally, the covariance-based parameterization of the closed-loop system can be rendered as 
\begin{align}
    x_{k+1} = (\bar{X}_1 - \bar{\Omega}_0) H x_k + \omega_k.  \label{eq:closed_db2}
\end{align}

\begin{lemma}[Closed-loop covariance for the indirect parameterization]
Consider the data-based closed-loop representation
\eqref{eq:closed_db2}. 
The distribution of the state $x_k$ at time $k$ for this representation is given by \( x_k \sim \mathcal{N}(\mu_k, \Sigma_k) \) with the update rules as
\begin{align}
\mu_{k+1} &= \bar{X}_1 H \mu_k, \label{eq:db_mean2} \\
\Sigma_{k+1} &= \bar{X}_1 H \Sigma_k H^\top \bar{X}_1^\top + \frac{\mathrm{Tr}(H \Sigma_k H^\top \Phi)}{N}W + W. \label{eq:db_cov2}
\end{align}
where the additional term $\frac{\mathrm{Tr}(H \Sigma_k H^\top \Phi)}{N}W$ arises from
propagating the uncertainty in the unknown noise samples $\Omega_0$ through
the data-based closed-loop map.
\end{lemma}

\begin{proof}
By taking the expectation from both sides of \eqref{eq:closed_db2}, the mean dynamics in \eqref{eq:db_mean2} results considering the independence of $\bar{\Omega}_0$, $H$, and $x_k.$ 
Expanding the second moment yields
\begin{align}
& \mathbb{E} [x_{k+1}  x_{k+1}^T] = \nonumber \\
& \mathbb{E} [((\bar{X}_1 - \bar{\Omega}_0) H x_k+\omega_{k})  ((\bar{X}_1 - \bar{\Omega}_0) H x_k+\omega_{k})^\top], \nonumber \\ 
& \Sigma_{k+1} = \bar{X}_1 H \mathbb{E} [x_k x_k^\top ] H^\top \bar{X}_1^\top  + \mathbb{E} [ \bar{\Omega}_0 H x_k x_k^\top H^\top \bar{\Omega}_0^\top] \nonumber \\
% & +\mathbb{E} [\omega_k \omega_k^\top]  \nonumber \\
%  &= \bar{X}_1 H \Sigma_k H^\top \bar{X}_1^\top  + \mathbb{E} [ \bar{\Omega}_0 H  \mathbb{E} [  x_k x_k^\top ] H^\top \bar{\Omega}_0^\top] +W \nonumber \\
 &= \bar{X}_1 H \Sigma_k H^\top \bar{X}_1^\top  + \mathbb{E} [ \bar{\Omega}_0 H  \Sigma_k H^\top \bar{\Omega}_0^\top] +W, \label{eq:covariancemid2}
\end{align}
and by substituting $\mathbb{E} [ \bar{\Omega}_0 H  \Sigma_k H^\top \bar{\Omega}_0^\top]  = \frac{\mathrm{Tr}(H \Sigma_k H^\top \Phi)}{N}W$ in \eqref{eq:covariancemid2}, the uncertainty-aware closed-loop data-driven covariance equation \eqref{eq:db_cov2} results. This completes the proof. 
\end{proof}

\begin{remark}
    Similar to Remark \ref{rem:cov1}, the data-based closed-loop covariance \eqref{eq:db_cov2} has an extra term $\frac{\mathrm{Tr}(H \Sigma_k H^\top \Phi)}{N}W$ compared to the one presented in~\cite{zhao2024data}. In~\cite{zhao2024data}, they ignored the term $\bar{\Omega}_0 H$ in \eqref{eq:ABK2}, but we treated this term as a random variable. This uncertainty propagates in the closed-loop covariance through the $\frac{\mathrm{Tr}(H \Sigma_k H^\top \Phi)}{N}W$ term. The resulting uncertainty-aware covariance matrix brings robustness into the design of the feedback control gain. 
\end{remark}
\begin{remark}
    The steady-state covariance matrices in \eqref{eq:cov}, \eqref{eq:db_cov1}, and \eqref{eq:db_cov2} can be obtained by replacing $\Sigma_{k+1}$ and $\Sigma_k$ with 
    $\Sigma.$
\end{remark}

\section{Natural Gradient Descent Perspective for Stochastic Control}
This section introduces foundational concepts underlying the proposed natural gradient-based control framework. 
%We first recall essential notions in gradient optimization, including the standard and natural gradient methods, and then formulate their application in the control of stochastic linear systems.

\subsection{Gradient-Based Optimization Fundamentals}
% The classical gradient descent (GD) update rule for a cost function $J(\theta)$ with respect to parameters $\theta$ is
% \begin{equation}
%     \theta_{k+1} = \theta_k - \alpha \nabla_\theta J(\theta_k),
% \end{equation}
% where $\alpha > 0$ is a constant step size. This formulation assumes a flat, Euclidean geometry in the parameter space, which may lead to inefficient convergence in problems with direction-dependent cost curvature.
%\subsection{Natural Gradient Descent}
Natural gradient descent (NGD) is a modified GD that accounts for the geometry of the parameter space
% It does so using a Riemannian metric defined by 
using the Fisher Information Matrix (FIM) as
\begin{equation}
    \theta_{k+1} = \theta_k - \alpha G(\theta_k)^{-1} \nabla_\theta J(\theta_k),
\end{equation}
where $G(\theta_k)$ is the FIM at iteration $k$ and serves as a curvature-aware preconditioner. This approach yields more stable and efficient learning trajectories, particularly in ill-conditioned spaces.

\subsection{Natural Gradient for Mean Dynamics}
Reference~\cite{esmzad2024gdl} introduced a closed-loop control framework that combines gradient descent with a linear control policy. Instead of shaping trajectories indirectly through cost-function tuning, it enforces a gradient-descent-like evolution on the system state by parameterizing the stable closed-loop dynamics as
\begin{equation} \label{25}
    A + BK = I - 2\Gamma P,
\end{equation}
where \( \Gamma \) and \( P \succ 0 \) serve as a step-size matrix and a Lyapunov matrix, respectively. This formulation ensures that the state evolution mimics a preconditioned gradient descent
\begin{equation} \label{26}
    x_{k+1} = x_k - \Gamma \nabla V(x_k), \quad \text{with} \quad V(x_k) = x_k^\top P x_k.
\end{equation}

In the stochastic control setting, we consider the expected cost $\mathbb{E}[V(x_k)]$ under $x_k \sim \mathcal{N}(\mu_k, \Sigma_k)$, where $V(x_k)$ is a positive-definite quadratic function. The NGD update for the state mean then becomes
\begin{equation}
    \mu_{k+1} = \mu_k - \alpha G(\mu_k)^{-1} \nabla_{\mu_k} \mathbb{E}[V(x_k)]. \label{eq:ngd}
\end{equation}
This update reflects the direction of steepest descent in a Riemannian manifold characterized by $G(\mu_k)$, thereby adjusting the trajectory based on the uncertainty geometry of the system.

\subsection{Fisher Information Matrix for Gaussian States}
Using a linear policy as in \eqref{eq:control}, the state $x_k$ is Gaussian-distributed; $x_k \sim \mathcal{N}(\mu_k, \Sigma_k)$. The FIM with respect to $\mu_k$ is given by
\begin{equation}
    G(\mu_k) = \mathbb{E}\left[ \nabla_{\mu_k} \log p(x_k) \nabla_{\mu_k} \log p(x_k)^\top \right]= \Sigma_k^{-1}.
    \label{eq:FIM}
\end{equation}

This reveals that the FIM equals the inverse of the state covariance, meaning the NGD update naturally adjusts for state uncertainty. {Directions with higher variance (lower information) receive larger updates, while directions with high confidence receive smaller updates.}

%Directions with higher variance (lower information) receive smaller updates, while directions with high confidence are updated more aggressively~\cite{esmzad2026natural}.

\section{Proposed Framework} \label{sec:framework}
Traditional control methods rely on explicit models and cost tuning, often requiring trial and error to achieve desired transient and robustness properties. In contrast, this work proposes a NGD framework that directly shapes state trajectories through geometric preconditioning and a scalar tuning parameter, without model identification. We develop the method in a direct data-driven setting, establish stability and convergence guarantees, and validate it on a real robotic platform.

\subsection{Theoretical results}
In this subsection, we present the main stability and synthesis guarantees of the proposed framework. 
\begin{theorem} \label{th:theorem1}
Consider the system \eqref{eq:syst} with closed-loop data-based parameterization \eqref{eq:closed_db1}.  Let Assumptions \ref{assu:AB}-\ref{assum:datarank} hold. 
Assume that $\alpha>0$ is given and $0<\lambda<1$. {If} $Y,\: F,\: \Sigma, \: M$, and $Z$ form a feasible solution to the following linear problem

\begin{subequations}
\label{eq:theorem1}
\begin{align}
    X_1 F &= Y - 2\alpha \Sigma, \label{eq:abg}\\
    \begin{bmatrix}
        \lambda Y & (Y-2\alpha \Sigma)^\top \\ * & Y
    \end{bmatrix} &\succeq 0,\label{eq:lmi_stability1}\\
     \begin{bmatrix}
    M & F \\ * & Z
    \end{bmatrix} &\succeq 0, \label{eq:lmi_M1} \\
    \begin{bmatrix}
        Z & Y\\ * & \Sigma
    \end{bmatrix} &\succeq 0 \label{eq:lmi_Z1}, \\
     X_1 M X_1^\top +  \mathrm{Tr}(M)W + W - \Sigma  & \preceq 0 ,  \label{eq:covlmi1}\\
     X_0 F &= Y, \label{eq:X0G1} 
\end{align}
\end{subequations}
{then, the preconditioned NGD control design renders the data-based closed-loop dynamics \eqref{eq:closed_db1}  $\lambda$-contractive} in expectation (and thus stable) with $K=U_0 F P$, $P=Y^{-1}$.
%then, the preconditioned NGD control in \eqref{eq:ngd} makes the data-based closed-loop dynamics \eqref{eq:closed_db1} with $K=U_0 F P$, $P=Y^{-1}$, $\lambda-$contractive in expectation (and thus stable).
\end{theorem}
\begin{proof}
Consider the following Lyapunov candidate 
\begin{align}
    V(x_k) = \mathbb{E}_{x_k \sim \mathcal{N}(\mu_k, \Sigma_k)}\left[x_k^\top P x_k\right], \quad P\succ 0.
\end{align}
and let $G(\mu_k)=\Sigma^{-1}$. As a result, the natural GD reads
\begin{equation}
   \mu_{k+1} = \mu_k - 2 \alpha \Sigma P \mu_k. \label{eq:ngd_ss}
\end{equation}
By selecting $\alpha$, one can ensure the natural GD is contractive and thus stable; i.e. $\rho(I-2\alpha \Sigma P)<1$. 
To design the controller gain $K$ from the preconditioned GD control design \eqref{eq:ngd_ss}, we introduce a new parametrization of the closed-loop dynamics as $X_1 G=I-{2}\alpha \Sigma P$, which is equivalent to \eqref{eq:abg} after multiplying both sides from right by $Y$ and using $F=GY$. {This ensures that the closed-loop dynamics in \eqref{eq:db_mean1} matches the NGD dynamics in \eqref{eq:ngd_ss}.} Now, we show that $K=U_0 F P$ makes the mean of the data-based closed-loop system $\lambda-$contractive.
The controller gain $K$ makes the closed loop system $\lambda$-contractive, if 
\begin{align*}
    (X_1 G)^{\top}P(X_1 G) \preceq \lambda P.
\end{align*}
Multiplying both sides of the above inequality with $Y$ from left and right and using $X_1 G=I-2\alpha \Sigma P$, 
one gets
\begin{align}
    (Y-2\alpha \Sigma)^\top P (Y-2\alpha \Sigma) &\preceq \lambda Y
\end{align}
{Using \(P=Y^{-1}\) and rearranging, this inequality becomes
\begin{align}
\lambda Y -(Y-2\alpha \Sigma)^\top Y^{-1} (Y-2\alpha \Sigma) \succeq 0 .
\end{align}
Since \(Y \succ 0\), by the Schur complement lemma, this is equivalent to
 \eqref{eq:lmi_stability1}.}
In the preconditioned GD control, we select the FIM as 
%the inverse of the closed-loop stationary covariance matrix; i.e. 
$G(\mu_k)=\Sigma^{-1}$. The stationary closed-loop covariance using the controller gain $K$ reads
\begin{equation}
   \Sigma = X_1 G \Sigma G^\top X_1^\top + \mathrm{Tr}(G \Sigma G^\top)W + W. \label{eq:db_cov1_ss}
\end{equation}
which couples the unknown variables $\Sigma$ and $G$ and thus should be incorporated in the design. {Introduce $M$ and $Z$ as parts of convex relaxation}
\begin{align}
    M &\succeq G \Sigma G^\top \label{eq:M_lmi} \\
    Z &\succeq Y \Sigma^{-1} Y. \label{eq:Z_lmi}
\end{align}
{Since $G=FY^{-1}$, the inequality in \eqref{eq:M_lmi} can be written as
\begin{align*}
    M &\succeq F Y^{-1} \Sigma Y^{-1} F^{\top}\succeq FZ^{-1} F^{\top}
\end{align*}
where we used $ Y^{-1} \Sigma Y^{-1} \succeq Z^{-1}$ from \eqref{eq:Z_lmi} to get the last inequality which can be written as \eqref{eq:lmi_M1}. The inequality in \eqref{eq:Z_lmi} {can} be written as \eqref{eq:lmi_Z1}. Based on \eqref{eq:M_lmi}, one can write
\begin{align*}
    X_1 M X_1^\top &+ \mathrm{Tr}(M)W + W \\
    & {\succeq} X_1 G \Sigma G^\top X_1^\top + \mathrm{Tr}(G \Sigma G^\top)W + W.
\end{align*}
{
Hence,
\(
X_1 M X_1^\top + \mathrm{Tr}(M)W + W
\)
is an upper bound on the stationary covariance expression
\(
X_1 G \Sigma G^\top X_1^\top + \mathrm{Tr}(G\Sigma G^\top)W + W .
\)
Therefore, a sufficient condition for
\[
X_1 G \Sigma G^\top X_1^\top + \mathrm{Tr}(G\Sigma G^\top)W + W \preceq \Sigma
\]
is
\[
X_1 M X_1^\top + \mathrm{Tr}(M)W + W \preceq \Sigma,
\]
which is equivalently to \eqref{eq:covlmi1}.}
} {From \eqref{eq:KD0G}, we have
 \(U_0G=K\) and \(X_0G=I\). Using the change of variables
\(F=GY\) and \(P=Y^{-1}\), it follows that
\[
X_0F = X_0GY = Y,
\]
which gives \eqref{eq:X0G1}. Moreover,
\[
K = U_0G = U_0FY^{-1} = U_0FP.
\]}
% which are equivalents to \eqref{eq:lmi_M1} and \eqref{eq:lmi_Z1} respectively, then 
% one can write \eqref{eq:db_cov1_ss} as \eqref{eq:covlmi1}. 
% Finally $X_0 G = I$ is equivalent to \eqref{eq:X0G1}. This completes the proof.
\end{proof}

\begin{theorem} \label{th:theorem2}
Consider the system \eqref{eq:syst} with closed-loop data-based parameterization \eqref{eq:closed_db2}.  Let Assumptions \ref{assu:AB}-\ref{assum:datarank} hold. 
Assume that $\alpha>0$ is given and $0<\lambda<1$. 
{If} $Y,\: F,\: \Sigma, \: M$, and $Z$ form a feasible solution to the following linear problem

\label{eq:theorem2}
\begin{align}
    \bar{X}_1 F &= Y - 2\alpha \Sigma,  \quad
     \bar{X}_1 M \bar{X}_1^\top + \frac{\mathrm{Tr}(M \Phi)}{N}W+ W - \Sigma  \preceq 0, \\
     \bar{X}_0 F &= Y, \quad
    \begin{bmatrix}
        \lambda Y & (Y-2\alpha \Sigma)^\top \\ * & Y
    \end{bmatrix} \succeq 0, \quad \nonumber \\
    &\begin{bmatrix}
    M & F \\ * & Z
    \end{bmatrix} \succeq 0,  \quad
    \begin{bmatrix}
        Z & Y\\ * & \Sigma
    \end{bmatrix} \succeq 0, \label{eq:thoe2:lmi2}
\end{align}
{then, the preconditioned NGD control design renders the data-based closed-loop dynamics \eqref{eq:closed_db2}  $\lambda$-contractive} in expectation (and thus stable) with $K=\bar{U}_0 H$, $H=F Y^{-1}$, $P=Y^{-1}$.
%then, the preconditioned natural GD control makes the data-based closed-loop dynamics \eqref{eq:closed_db1} with $K=\bar{U}_0 H$, $H=F Y^{-1}$, $P=Y^{-1}$, $\lambda-$contractive in expectation (and thus stable).
\end{theorem}
\begin{proof}
    The proof is similar to the proof of Theorem 1 and hence it is omitted here.
\end{proof}
% {add control and state constrains as well. The state constraint can be $Y - x_kx_k^\top < 0$. also we can add obstacles for trajectory tracking or planning.}

\begin{remark}
\textbf{(Comparison of Theorems 1 and 2)}\quad
Theorems~1 and~2 provide two alternative data-driven formulations for ensuring \(\lambda\)-contractiveness. Theorem~1 uses a direct parameterization based on raw input-state data, while Theorem~2 uses a covariance-based parameterization built from second-order statistics. The latter is typically more memory-efficient and better suited to streaming or large-scale data, whereas the former can be more precise for smaller datasets. A summary of both formulations is given in Algorithm~\ref{alg:ngd-dd} in the Appendix.
\end{remark}

{
\begin{remark}\textbf{(Optimality of GD design)}\quad
    Theorems~\ref{th:theorem1} and \ref{th:theorem2} guarantee \(\lambda\)-contractiveness for the proposed data-driven controllers. Their main goal is to design a controller gain that reproduces GD dynamics. The connection to optimality and to the LQR weights \(Q\) and \(R\) is discussed in \cite{esmzad2026natural}, where it is shown that many closed-loop behaviors typically obtained by tuning \(Q\) and \(R\) can instead be recovered in our framework through a single scalar step size \(\alpha\). This preserves the connection to classical LQR design while providing a more compact parameterization.
\end{remark}}

We provide an explicit iteration complexity (time-to-$\varepsilon$) bound implied by the NGD design, with $G(\mu_k)=\Sigma^{-1}$ as in (26) and the contraction enforced by Theorems~1–2.

\begin{lemma}\label{rate}
Suppose the hypotheses of Theorem~1 (or Theorem~2) hold, yielding a feasible solution with $\lambda\in(0,1)$ and stationary covariance $\Sigma\succ0$. Then, for any $\varepsilon>0$,
\[
k \;\ge\; \frac{2}{\,1-\lambda\,}\,\log\!\Big(\frac{\|\mu_0\|_{P}}{\varepsilon}\Big)
\quad\Longrightarrow\quad
\|\mu_k\|_{P}\le \varepsilon.
\]
{Moreover, if $\alpha$ is chosen so that $ 0 < \alpha < \frac{1}{\lambda_{max}(\Sigma P)}$ then one can select $\lambda\le 1-2\alpha\,\lambda_{\min}(\Sigma P)$.}
\end{lemma}

\begin{proof}
By \eqref{25}–\eqref{26}, the mean dynamics under the NGD controller evolve as
\[
\mu_{k+1} = (I - 2\alpha\,\Sigma P)\,\mu_k, 
\qquad P = Y^{-1}\succ0,
\]
where $\alpha>0$ is the step size and $\Sigma,P\succ0$ are symmetric matrices.
Define $A_c := I - 2\alpha\,\Sigma P$ and the $P$-weighted similarity transform
$A_P := P^{1/2} A_c P^{-1/2}$ so that
\(
\|\mu_{k+1}\|_P^2 = {(P^{1/2}\mu_k)^\top A_P^\top A_P (P^{1/2}\mu_k)}.
\)
As a result
\begin{align*}
    \Vert \mu_{k+1} \Vert_{P}^2 \leq \rho(A_P^\top A_P) \Vert \mu_{k} \Vert_{P}^2 =\rho(A_P)^2 \Vert \mu_{k} \Vert_{P}^2.
\end{align*}
Since $A_P = I - 2\alpha P^{1/2}\Sigma P^{1/2}$ is symmetric
\begin{align*}
    \rho(A_p)=\max_{i} \vert 1-2\alpha \lambda_i(P^{1/2}\Sigma P^{1/2})\vert.
\end{align*}
To have $\rho(A_p)<1$,  one can select the learning rate to satisfy $0 < \alpha < \frac{1}{\lambda_{max}(P^{1/2}\Sigma P^{1/2})}$. Since $P^{1/2}\Sigma P^{1/2}$ is similar to  $\Sigma P$, one gets
\(
    0 < \alpha < \frac{1}{\lambda_{max}(\Sigma P)}.
\)
Thus, it is enough to select the contraction factor $\lambda$ as
\(
    \lambda \leq \rho(A_P)=\max_{i} \vert 1-2\alpha \lambda_i(P^{1/2}\Sigma P^{1/2})\vert
    =1-2\alpha \lambda_{\min}(\Sigma P).
\)
Combining this with the expression above yields the admissible range stated in the lemma. Iterating $\mu_{k+1}^\top P\mu_{k+1}\le \lambda\,\mu_k^\top P\mu_k$
gives
\(
\|\mu_{k}\|_{P}\le \lambda^{k/2}\|\mu_0\|_P,
\)
and for any $\varepsilon>0$,
\(
k \ge \tfrac{2}{1-\lambda}\log\!\tfrac{\|\mu_0\|_P}{\varepsilon}
\)
guarantees $\|\mu_k\|_P\le \varepsilon$.
\end{proof}

\begin{remark}
Lemma~\ref{rate} gives the iteration complexity explicitly: the NGD controller reaches an \(\varepsilon\)-ball in \(O\!\big(\frac{1}{1-\lambda}\log(1/\varepsilon)\big)\) steps. It also quantifies how \(\alpha\) affects \(\lambda\), consistent with the design intuition; see \eqref{eq:FIM},  \eqref{eq:lmi_stability1}, and \eqref{eq:thoe2:lmi2}. \end{remark}

The following lemma and remark quantify the sample size needed for Assumption~3 to hold with high probability under Gaussian disturbances and persistently exciting inputs. 

{\begin{lemma}\label{lem:rich-phi}
Let \(D_0\) in Assumption~3 be formed by \(N\) columns of
$z_k=[u_k^T \quad x_k^T]^T \in \mathbb{R}^{m + n}$. 
Let the process noise be i.i.d.\ mean-zero Gaussian. Suppose that the
hypotheses of Theorem~1 (or Theorem~2) hold, let \(u_k=Kx_k\), where
\(K\) is the controller returned by Theorem~1 (or Theorem~2). 
Let
\(
\Sigma_D:=\mathbb E[z_k z_k^\top]\succ0.
\)
Then, there exist system-dependent constant \(C>0\) such that, for any
\(\delta\in(0,1)\) and any \(\varepsilon\in(0,1)\), if
\begin{equation}
N \ge C\,\frac{m+n+\log(1/\delta)}{\varepsilon^2},
\label{eq:N-conc}
\end{equation}
then, with probability at least \(1-\delta\),
\begin{align}
\|\Phi-\Sigma_D\|
&\le
\varepsilon \|\Sigma_D\|,
\label{eq:cov-dev}
\\
\lambda_{\min}(\Phi)
&\ge
\bigl(1-\varepsilon \kappa(\Sigma_D)\bigr)\lambda_{\min}(\Sigma_D).
\label{eq:lmin-dev}
\end{align}
where $\kappa(\Sigma_D)$ is the condition number of $\Sigma_D$.
\end{lemma}} 

{\begin{proof}
By Theorem~1 (or Theorem~2), the designed controller enforces the NGD
mean dynamics
\(
\mu_{k+1}
=
(I-2\alpha\Sigma P)\mu_k
=
A_{\mathrm{NGD}}\mu_k,
\)
where \(A_{\mathrm{NGD}}:=I-2\alpha\Sigma P\).
By Lemma \ref{rate}, we have
\(
\rho(A_{\mathrm{NGD}})<1.
\)
Hence, \(A_{\mathrm{NGD}}\) is Schur stable. \\
\indent Since the process noise is i.i.d.\ Gaussian and the closed-loop covariance
recursion admits the stationary solution \(\Sigma\) from Theorem~1 (or
Theorem~2), the state process is a stationary Gaussian linear process
satisfying
\begin{equation}
x_{k+1}=A_{\mathrm{NGD}}x_k+\omega_k,
\qquad
x_k\sim\mathcal N(0,\Sigma).
\label{eq:x-stable-recursion}
\end{equation}
In particular, \(\{x_k\}\) is mean-zero and strictly stationary. Now define
\[
C:=\begin{bmatrix}K\\I\end{bmatrix}\in\mathbb R^{(m+n)\times n},
\,
J:=\begin{bmatrix}0_{n\times m}&I_n\end{bmatrix}\in\mathbb R^{n\times (m+n)}.
\]
Then,
\[
z_k=\begin{bmatrix}u_k\\x_k\end{bmatrix}
=
\begin{bmatrix}K\\I\end{bmatrix}x_k
=
Cx_k,
\qquad
x_k=Jz_k.
\]
Using \eqref{eq:x-stable-recursion}, we obtain
\begin{align}
z_{k+1}
&=
Cx_{k+1}
=
C(A_{\mathrm{NGD}}x_k+\omega_k) \notag\\
&=
CA_{\mathrm{NGD}}J z_k + C\omega_k .
\label{eq:z-recursion-derived}
\end{align}
Hence, \(\{z_k\}\) satisfies a linear recursion of the form
\[
z_{k+1}=A_z z_k+\xi_k,
\qquad
A_z:=CA_{\mathrm{NGD}}J,
\qquad
\xi_k:=C\omega_k .
\]
Because \(\{\omega_k\}\) is i.i.d.\ mean-zero Gaussian, the sequence
\(\{\xi_k\}\) is also i.i.d.\ mean-zero Gaussian, hence sub-Gaussian. \\
We now show that \(A_z\) is stable. Since \(JC=I_n\), for every integer
\(\ell\ge1\),
\(
A_z^\ell = C A_{\mathrm{NGD}}^\ell J.
\)
Therefore,
\(
\|A_z^\ell\|
\le
\|C\|\,\|A_{\mathrm{NGD}}^\ell\|\,\|J\|.
\)
Because \(A_{\mathrm{NGD}}\) is Schur stable, there exist constants
\(M>0\) and \(r\in(0,1)\) such that
\(
\|A_{\mathrm{NGD}}^\ell\|\le Mr^\ell,
\quad \forall \ell\ge0.
\)
Thus,
\(
\|A_z^\ell\|
\le
\|C\|\,M\,\|J\|\,r^\ell,
\)
which implies \(\rho(A_z)<1\). Hence, the induced regressor process
\(\{z_k\}\) satisfies a stable linear recursion. \\
\indent Since \(z_k=Cx_k\) and \(x_k\sim\mathcal N(0,\Sigma)\), the regressor process is a stationary Gaussian with covariance
\(
\Sigma_D
=
\mathbb E[z_k z_k^\top]
=
C\,\mathbb E[x_kx_k^\top]\,C^\top
=
C\Sigma C^\top.
\)
Using \(z_k=Cx_k\), we have
\[
D_0=CX_0,
\qquad
\Phi=\frac1N D_0D_0^\top
= C\left(\frac1N X_0X_0^\top\right)C^\top,
\]
and
\(
\Sigma_D=C\Sigma C^\top .
\)
Hence,
\[
\Phi-\Sigma_D
=
C\left(\frac1N X_0X_0^\top-\Sigma\right)C^\top,
\]
so that
\begin{equation}
\|\Phi-\Sigma_D\|
\le
\|C\|^2
\left\|
\frac1N X_0X_0^\top-\Sigma
\right\|.
\label{eq:Phi-from-X}
\end{equation}
For the stable linear state process \(\{x_k\}\), the covariates-matrix concentration machinery in Section~IV-B and Theorem~5 of \cite{jedra2022finite} yields a two-sided spectral concentration for \(X_0\). In particular, this implies a non-asymptotic concentration of the empirical state covariance \(\frac1N X_0X_0^\top\) around the stationary covariance \(\Sigma\). Therefore, there exist constants \(C_1,C_2>0\) such that, for all \(t>0\),
\[
\Pr\!\left(
\left\|
\frac1N X_0X_0^\top-\Sigma
\right\|
\ge
C_1\|\Sigma\|\sqrt{\frac{n+m+t}{N}}
\right)\le e^{-t}.
\]
Combining this with \eqref{eq:Phi-from-X} yields
\[
\Pr\!\left(
\|\Phi-\Sigma_D\|
\ge
C_1\|C\|^2\|\Sigma\|\sqrt{\frac{n+m+t}{N}}
\right)\le e^{-t}.
\]
Since \(\Sigma_D=C\Sigma C^\top\), this gives \eqref{eq:cov-dev} after absorbing the fixed matrix factor into the constants, when \eqref{eq:N-conc} is satisfied. \\
Besides, assuming \(\Sigma_D\succ0\), Weyl's inequality gives
\[
\lambda_{\min}(\Phi)
=
\lambda_{\min}\bigl(\Sigma_D+(\Phi-\Sigma_D)\bigr)
\ge
\lambda_{\min}(\Sigma_D)-\|\Phi-\Sigma_D\|.
\]
Combining this with \eqref{eq:cov-dev} yields
\(
\lambda_{\min}(\Phi)
\ge
\lambda_{\min}(\Sigma_D)-\varepsilon\lambda_{\max}(\Sigma_D)
=
\bigl(1-\varepsilon\kappa(\Sigma_D)\bigr)\lambda_{\min}(\Sigma_D),
\)
which proves \eqref{eq:lmin-dev}.
\end{proof}}

{The following remark records an alternative lower-tail richness guarantee based on the block martingale small-ball (BMSB) condition.}

{\begin{definition}[BMSB condition {\cite{simchowitz2018learning}}]
\label{def:bmsb}
Let \(\{\mathcal F_k\}_{k\ge 0}\) be a filtration. A sequence
\(\{z_k\}_{k\ge 0}\subset \mathbb R^{m+n}\) satisfies the BMSB condition \((k_0,\Gamma_{\mathrm{sb}},p)\) if
there exist an integer \(k_0\ge 1\), a symmetric matrix
\(\Gamma_{\mathrm{sb}}\succ 0\), and a constant \(p\in(0,1]\) such that,
for every unit vector \(v\in\mathbb S^{m+n-1}\) and every \(k\ge 0\),
\begin{equation}
\frac{1}{k_0}\sum_{i=k}^{k+k_0-1}
\Pr\!\left(
|v^\top z_i|
\ge
\sqrt{v^\top \Gamma_{\mathrm{sb}} v}
\,\middle|\, \mathcal F_{k-1}
\right)\ge p .
\label{eq:bmsb}
\end{equation}
\end{definition}}

{\begin{remark}
If \(\{z_k\}\) satisfies Definition~\ref{def:bmsb}, then there exist
constants \(c_{\mathrm{sb}},C_{\mathrm{sb}}>0\), depending only on
\(k_0\) and \(p\), such that, for any \(\delta\in(0,1)\), if
\begin{equation}
N \ge C_{\mathrm{sb}}\bigl(m+n+\log(1/\delta)\bigr),
\label{eq:N-bmsb}
\end{equation}
then, with probability at least \(1-\delta\),
\begin{equation}
\lambda_{\min}(\Phi)\ge
c_{\mathrm{sb}}\,\lambda_{\min}(\Gamma_{\mathrm{sb}})>0 .
\label{eq:bmsb-lower}
\end{equation}
To see this, fix any unit vector \(v\in\mathbb R^{m+n-1}\) and define the scalar process
\(
y_k(v):=v^\top z_k.
\)
From Definition~\ref{def:bmsb}, for every \(k\ge0\),
\[
\frac{1}{k_0}\sum_{i=k}^{k+k_0-1}
\Pr\!\left(
|y_i(v)|
\ge
\sqrt{v^\top\Gamma_{\mathrm{sb}}v}
\,\middle|\,\mathcal F_{k-1}
\right)\ge p.
\]
Thus, uniformly over directions \(v\), the projected process has a
nontrivial small-ball probability at the scale
\(\sqrt{v^\top\Gamma_{\mathrm{sb}}v}\).\\
Applying the BMSB lower-tail argument to the scalar
process \(\{y_k(v)\}\), and then using a standard finite-net argument over
\(\mathbb S^{m+n-1}\), yields the existence of constants
\(c_{\mathrm{sb}},C_{\mathrm{sb}}>0\), depending only on \(k_0\) and \(p\),
such that whenever \eqref{eq:N-bmsb} holds, with probability at least
\(1-\delta\),
\begin{equation}
\frac{1}{N}\sum_{k=0}^{N-1}(v^\top z_k)^2
\ge
c_{\mathrm{sb}}\,v^\top\Gamma_{\mathrm{sb}}v.
\label{eq:uniform-small-ball}
\end{equation}
Since
\[
v^\top\Phi v
=
v^\top\!\left(\frac{1}{N}\sum_{k=0}^{N-1} z_k z_k^\top\right)\!v
=
\frac{1}{N}\sum_{k=0}^{N-1}(v^\top z_k)^2,
\]
\eqref{eq:uniform-small-ball} implies
\(
v^\top\Phi v
\ge
c_{\mathrm{sb}}\,v^\top\Gamma_{\mathrm{sb}}v,
\qquad
\forall v\in\mathbb S^{m+n-1},
\)
Taking the minimum over all unit
vectors \(v\) gives \eqref{eq:bmsb-lower}. Because
\(\Gamma_{\mathrm{sb}}\succ0\) and \(c_{\mathrm{sb}}>0\), we obtain
\(\Phi\succ0\). Finally, since \(\Phi=\frac1N D_0D_0^\top\), positivity of
\(\Phi\) implies that \(D_0\) has full row rank, i.e.,
\(
\operatorname{rank}(D_0)=m+n.
\)
\end{remark}}

% \begin{remark}
% Lemma~\ref{lem:rich-phi} calibrates the additional uncertainty terms that appear in covariance-based parametrization (cf. (14), (18)–(20)) and justifies that the trace inflation $\mathrm{Tr}(H\Sigma H^\top \Phi)$ concentrates around its mean at the stated $\tilde O\big(\sqrt{(m{+}n)/N}\big)$ rate. \end{remark}

\subsection{Algorithmic Implementation of Data-Driven NGD Controllers}
A step-by-step summary of the implementation procedure for both Theorem 1 and Theorem 2 is provided in Algorithm~\ref{alg:ngd-dd}.

% \begin{algorithm}[h]
% \caption{Direct Data-Driven Natural Gradient Control}
% \label{alg:ngd-dd}
% \begin{algorithmic}[1]
% \REQUIRE Collected input-state data matrices $U_0, X_0, X_1$, noise covariance $W$ or tuning proxy, NGD step size $\alpha > 0$, contraction factor $\lambda \in (0,1)$
% \ENSURE Feedback gain $K$
% \STATE Compute data-richness matrix $D_0 = \begin{bmatrix} U_0^\top & X_0^\top \end{bmatrix}^\top$ and verify $\text{rank}(D_0) = m + n$
% \STATE Choose parameterization strategy:
% \begin{itemize}
%     \item \textbf{Direct:} Use $X_1$, $X_0$, $U_0$ to define $G$ via $D_0 G = \begin{bmatrix} K^\top & I \end{bmatrix}^\top$
%     \item \textbf{Covariance-based:} Use the sample covariance $\Phi$ and $\bar{X}_1, \bar{X}_0,$ and $\bar{U}_0$ matrices to parametrize $H$ as $\Phi H = \begin{bmatrix} K^\top & I \end{bmatrix}^\top$
% \end{itemize}
% \STATE Formulate the corresponding LMIs for either Theorem 1 or Theorem 2
% \STATE Solve the resulting SDPs to obtain $F, Y, \Sigma$ and recover $P = Y^{-1}$
% \STATE Compute feedback gain $K = U_0 F P$ (or $K = \bar{U}_0 F P$ for covariance-based formulation)
% \RETURN $K$
% \end{algorithmic}
% \end{algorithm}

\begin{algorithm}[h]
\caption{Direct Data-Driven NGD Control}
\label{alg:ngd-dd}
\begin{algorithmic}[1]
\REQUIRE \(U_0,X_0,X_1,W,\alpha,\lambda\)
\ENSURE Feedback gain \(K\)
\STATE Form \(D_0=\begin{bmatrix}U_0^\top & X_0^\top\end{bmatrix}^\top\) and check \(\operatorname{rank}(D_0)=m+n\)
\STATE Select either:
\STATE \hspace{1em}\textbf{Direct:} \(D_0G=\begin{bmatrix}K\\I\end{bmatrix}\)
\STATE \hspace{1em}\textbf{Covariance-based:} \(\Phi H=\begin{bmatrix}K\\I\end{bmatrix}\)
\STATE Build the LMIs in Theorem~1 or Theorem~2 and solve for \(F,Y,\Sigma\)
\STATE Set \(P=Y^{-1}\) and recover \(K=U_0FP\) (or \(K=\bar U_0FP\))
\RETURN \(K\)
\end{algorithmic}
\end{algorithm}

\section{Simulation and Implementation} \label{sec:sims}
This section performs a comprehensive set of experiments in the Gazebo physics-based simulator and on a physical ROSbot XL robot. Comparative evaluations with Direct Data-Driven LQR (DDLQR)~\cite{ESMZAD2025112197} and model-based LQR are also conducted to illustrate the practical advantages and limitations of each method. The details regarding the robot platform and data collection procedure are given in Appendix \ref{app:robot}.

\subsection{Simulation and Experiment Setup}
To validate our controllers in both a real-world setup and a simulation setting based on physics, we use the ROSbot XL~\footnote{\url{https://husarion.com/manuals/rosbot-xl/}} and its model in the Gazebo simulator~\footnote{\url{https://gazebosim.org/}}, which provides full 3D physics. 
%The controllers are implemented in ROS2 and applied via velocity commands. The experiment and simulation proceed as follows.

\noindent \textbf{Data Collection:} Random piecewise-constant wheel velocity inputs are applied every $\Delta t = 0.5$ seconds, to collect $N=24$ data samples of the robot’s pose $x_k = [x, y, \phi]^\top$ and control inputs $u_k = [\omega_1, \omega_2, \omega_3, \omega_4]^\top$ in $U_0$, $X_0$, and $X_1$ matrices. Rotational inputs are constrained to keep $\phi$ small to validate our assumption of linear behavior of the robot as \eqref{eq:linear_model}. Appendix \ref{App:data_collection} details the data collection procedure.

\noindent \textbf{Controller Synthesis:} Using the collected data, we solve the SDPs from Theorem 1 and Theorem 2 to compute the controller gains $K$ for different values of $\alpha$. We used CVXPY~\cite{diamond2016cvxpy,agrawal2018rewriting} for modeling the convex optimization problems and MOSEK~\cite{mosek} as the solver.

\noindent \textbf{Planning:} The robot is tasked to go from an initial point $(x=0, y=0, \phi = 0)$ to a final target $(x=2.7, y=1.5, \phi = 0)$.

\subsection{Monte Carlo Analysis of Controller Robustness (Theorem 1)}

To assess the NGD controller's robustness (Theorem 1), we ran 20 Monte Carlo simulations under Gaussian noise. Comparing step sizes reveals a clear aggressiveness-robustness trade-off (Table~\ref{tab:mc_alpha_comparison}): $\alpha = 10^{-5}$ yields smooth, low-variance responses (Figs.~\ref{fig:mc_xy_phi_small}--\ref{fig:mc_controls_small}), whereas $\alpha = 0.2$ accelerates convergence but increases state and control variability (Figs.~\ref{fig:mc_xy_phi_large}--\ref{fig:mc_controls_large}, Appendix~\ref{app:monte}).

\begin{table}[H]
\centering
\vspace{-0.3cm}
\caption{Monte Carlo Summary Statistics Comparing two Step Sizes \( \alpha \) for Theorem 1}
\label{tab:mc_alpha_comparison}
\renewcommand{\arraystretch}{1.2}
\begin{tabular}{lcc}
\toprule
\textbf{Metric} & \( \alpha = 10^{-5} \) & \( \alpha = 0.2 \) \\
\midrule
Final \( x_T \) (mean ± std) & 2.700 ± 0.034 & 2.702 ± 0.111 \\
Final \( y_T \) (mean ± std) & 1.500 ± 0.046 & 1.503 ± 0.108 \\
Final \( \phi_T \) (mean ± std) & 0.000 ± 0.021 & 0.002 ± 0.096 \\
Avg. wheel std (\( \omega_1 \)–\( \omega_4 \)) & 0.043 & 0.207 \\
\bottomrule
\end{tabular}
\vspace{-0.3cm}
\end{table}

\subsection{Tuning Comparison: NGD vs. DDLQR} \label{subsec:behavior_sim}
This section compares the tuning of our NGD controllers (Theorems 1 and 2) via the step-size $\alpha$ against Direct Data-Driven LQR (DDLQR) via weighting matrices $Q$ and $R$ (details in Appendix~\ref{app:sim}).

For Theorem 1, tuning $\alpha$ dictates a clear trade-off between responsiveness and smoothness (Figs.~\ref{fig:dd1_states}-\ref{fig:dd1_controls}). Larger $\alpha$ accelerates convergence but risks aggressive, oscillatory inputs and actuator saturation; smaller $\alpha$ yields conservative, gradual trajectories. Theorem 2 exhibits similar trends but is significantly more sensitive and harder to tune (Figs.~\ref{fig:dd2_states}-\ref{fig:dd2_controls}). Its direct reliance on the sample covariance matrix makes it less robust under limited data.

Conversely, DDLQR relies on shaping $Q$ and $R$ (Figs.~\ref{fig:ddlqr_states}-\ref{fig:ddlqr_controls}, Table~\ref{tab:qr-tuning-sets}). While expressive, its tuning is unintuitive and highly coupled; in our tests, behavior remained largely similar across 15 distinct $(Q, R)$ pairs, requiring extensive trial-and-error to systematically adjust performance.

Overall, Theorem 1 offers a highly interpretable, geometry-aware tuning mechanism using a single scalar $\alpha$. It provides a straightforward intensity-smoothness trade-off that is significantly easier to tune than traditional DDLQR cost shaping, particularly in low-data regimes.

\subsection{Implementation on Physical Platform} \label{subsec:behavior_real}
To evaluate real-world practicality, we deployed the controllers on a physical ROSbot XL and its Gazebo simulator counterpart (see Appendix~\ref{app:real} for details). After design, we deployed the policies on the physical robot and cross-validated them in the Gazebo simulator. Real-world trajectories and control inputs are shown in Figs.~\ref{fig:real_states} and~\ref{fig:real_controls}, with corresponding simulation results in Figs.~\ref{fig:gazebo_states} and~\ref{fig:gazebo_controls}. 

Comparing the environments reveals distinct reliability profiles. Both standard model-based LQR and the controller from Theorem 1 demonstrated consistent, stable performance across simulation and hardware, validating Theorem 1's robustness to limited data and real-world uncertainties. Conversely, Theorem 2 proved difficult to tune and less consistent on hardware; its high sensitivity to hyperparameters ($\alpha$, $\lambda$) makes it fragile under unmodeled disturbances, likely exacerbated by the small sample size. Finally, DDLQR exhibited notable sim-to-real variance, though it surprisingly achieved smoother trajectories and lower control effort on the physical hardware than in simulation.

% In your document
\begin{figure}[h]
    \centering

    % Top row: Real-world results
    \begin{subfigure}[t]{0.45\textwidth}
        \centering
        \includegraphics[width=\linewidth]{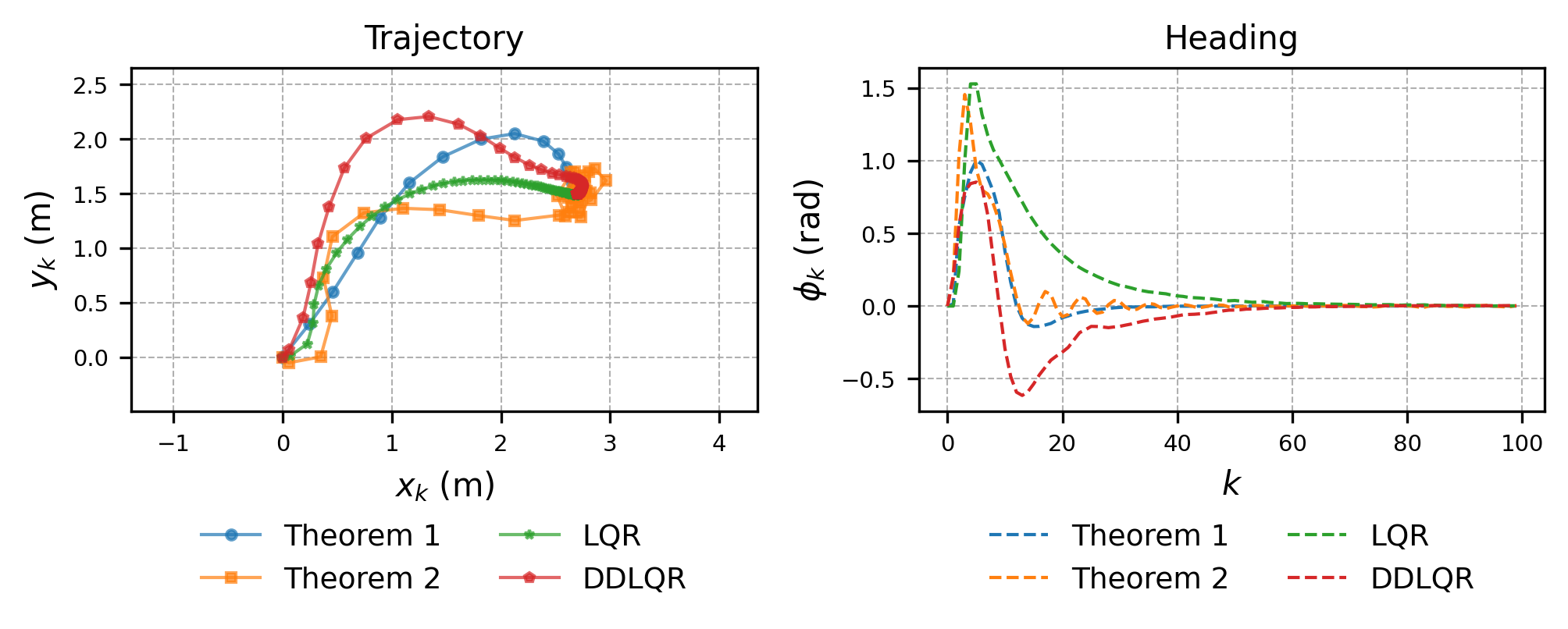}
        \caption{Real-world trajectory and heading.}
        \label{fig:real_states}
    \end{subfigure}
    \hfill
    \begin{subfigure}[t]{0.45\textwidth}
        \centering
        \includegraphics[width=\linewidth]{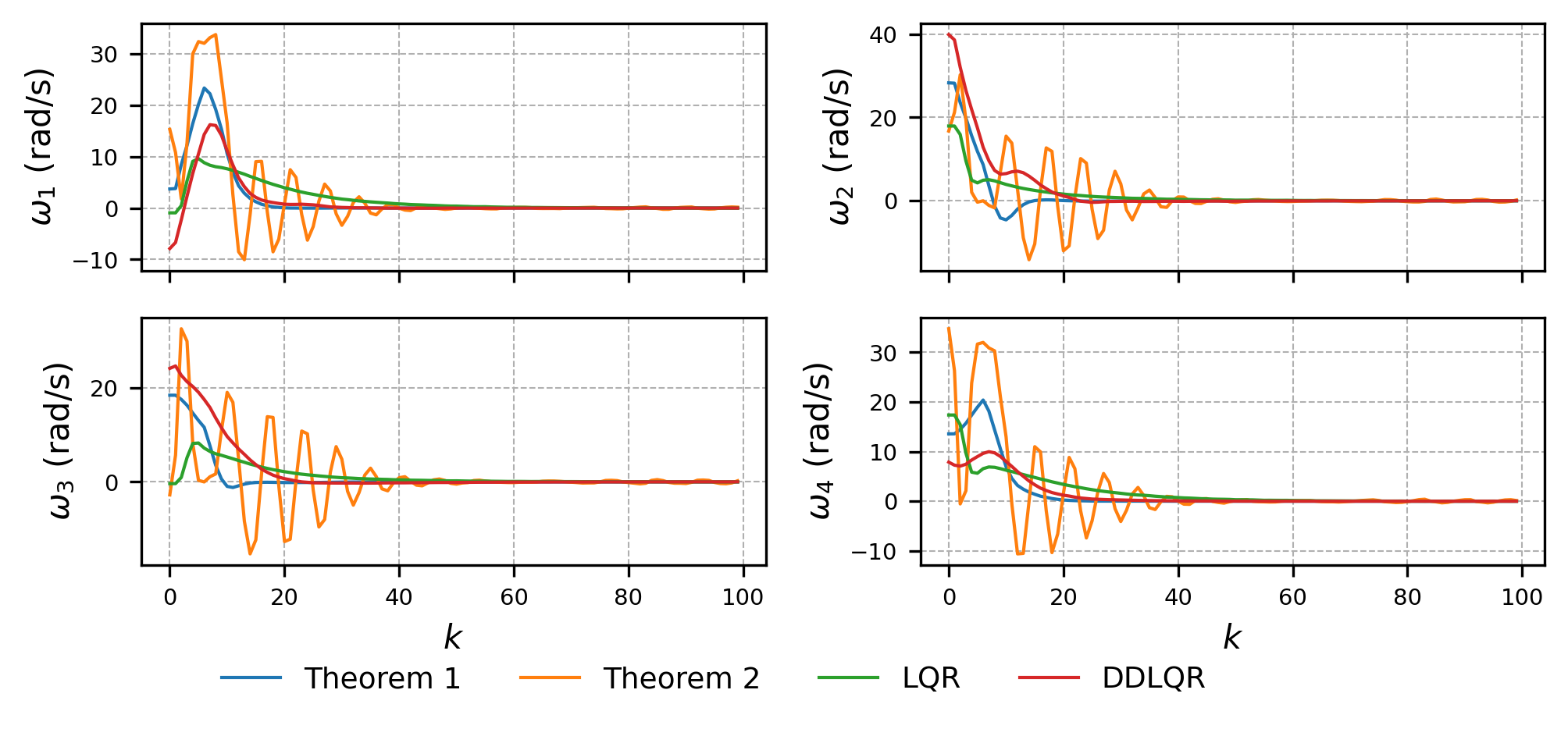}
        \caption{Wheel inputs in real-world deployment.}
        \label{fig:real_controls}
    \end{subfigure}

    % \vspace{0.1em}

    % Bottom row: Gazebo simulation
    \begin{subfigure}[t]{0.45\textwidth}
        \centering
        \includegraphics[width=\linewidth]{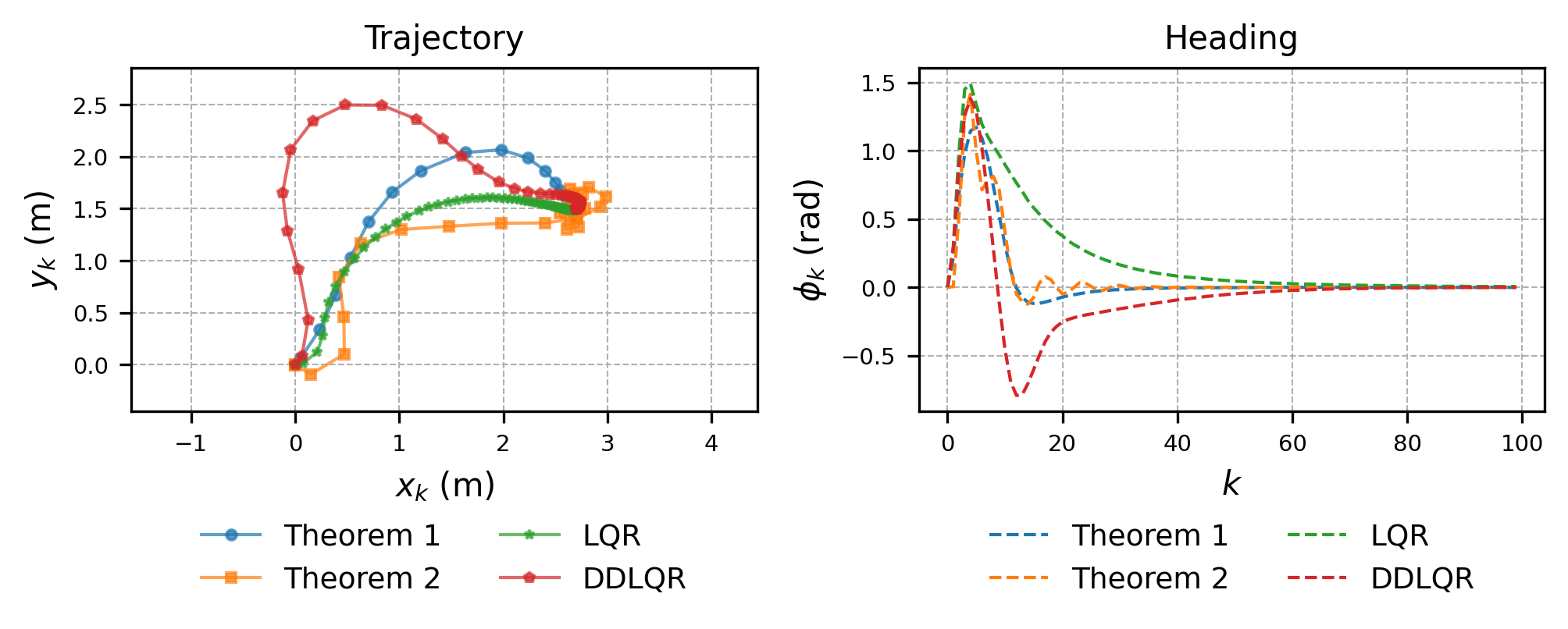}
        \caption{Simulated trajectory and heading in Gazebo.}
        \label{fig:gazebo_states}
    \end{subfigure}
    \hfill
    \begin{subfigure}[t]{0.45\textwidth}
        \centering
        \includegraphics[width=\linewidth]{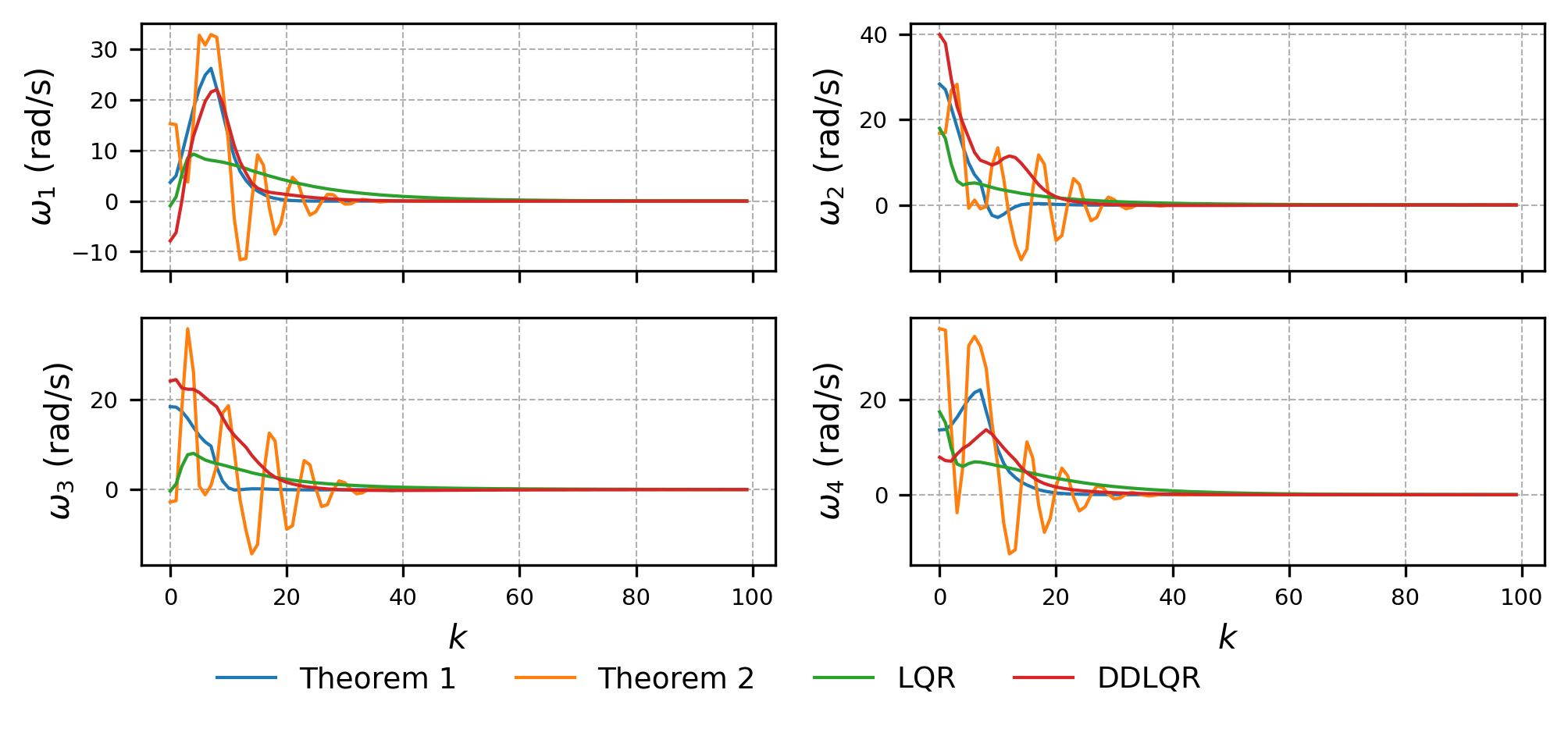}
        \caption{Simulated wheel inputs in Gazebo.}
        \label{fig:gazebo_controls}
    \end{subfigure}

    \caption{Comparison of controller behavior on the real robot (top row) and in Gazebo simulation (bottom row). Theorem 1 shows consistent performance across physical and simulated platforms.}
    \label{fig:real_vs_sim}
    \vspace{-0.4cm}
\end{figure}

{
\subsection{Robustness Analysis Under Increasing Noise Intensity}
\label{subsec:robustness_snr}
To analyze robustness, we perform a signal-to-noise (SNR)-based Monte Carlo study comparing five controllers (NGD from Theorems~1 and~2, DDLQR~\cite{ESMZAD2025112197}, LCLQR~\cite{DEPERSIS2021109548}, and model-based LQR) using initial data of length $N=55$.  We run $N_m=50$ trials over a horizon $T=504$ under varying process noise $W(\sigma) = \sigma^2 \mathrm{diag}(1,1,5)$ for $\sigma \in \{0.001, 0.007, 0.04, 0.1\}$, $Q=100I_3$, and $R=0.1I_4$. Performance is evaluated against the ratio $\mathrm{SNR(dB)} = 10 \log_{10}(P_{\mathrm{sig}}/\sigma^2)$, where $P_{\mathrm{sig}}$ is the average squared state magnitude from a noiseless closed-loop rollout. We track position RMSE, final position error, average control energy $\mathbb{E}[\|u\|^2]$, peak wheel-speed, input variation $\mathrm{RMS}(\Delta u)$, and time-of-arrival (TOA) to the target (Table~\ref{tab:snr_robustness}).
Decreasing SNR expectedly degrades terminal errors and TOA for all methods. At high SNR ($\sigma=0.001$), all controllers achieve similar accuracy ($\approx 10^{-3}$ m). At intermediate noise ($\sigma=0.04$), Theorem~1 yields the smallest RMSE among NGD variants and lower energy growth than data-driven baselines. Under severe noise ($\sigma=0.1$), NGD controllers remain stable under saturation, but their convergence slows significantly (TOA $\approx 90$ for Theorem~1, $\approx 104$ for Theorem~2) compared to LCLQR/DDLQR ($< 51$). This highlights a clear trade-off: Theorem~1 provides better bounded accuracy and lower control energy, while LCLQR/DDLQR achieve faster arrival under severe noise at the cost of higher input variation. {Finally, Table~2 does not present an \(\alpha\)-sweep within a given method. Rather, it compares different controllers under increasing noise levels, using one fixed representative \(\alpha\) for each NGD controller (Theorem~1 uses \(\alpha=0.01\), while Theorem~2 uses \(\alpha=10^{-5}\)). The effect of varying \(\alpha\) within each theorem is studied separately in Section~V.F through the \((N,\alpha)\)-sweep results reported in Tables~3 and~4. Those results show that \(\alpha\) acts as a direct tuning parameter for aggressiveness: increasing \(\alpha\) generally accelerates contraction and reduces TOA, but also increases control effort, input variation, and sensitivity to noise.}
}

\begin{table}[htbp]
\centering
\caption{{SNR-based robustness study on nonlinear robot model ($T=504, N_m=50$).}}
\label{tab:snr_robustness}
\setlength{\tabcolsep}{3pt}
\small
\begin{tabular}{l c cccc c}
\toprule
\textbf{Controller} & \boldmath$\sigma$ \textbf{(dB)} & \textbf{RMSE} & \textbf{F.Err} & \boldmath$\mathbb{E}[\|u\|^2]$ & \textbf{RMS} & \textbf{TOA} \\
 & & & & & \boldmath$(\Delta u)$ & \\
\midrule
Thm. 1 & 0.001 (46) & 0.195 & 0.002 & 23.15 & 1.39 & 4.0 \\
($\alpha=0.01$)  & 0.007 (29) & 0.196 & 0.015 & 23.35 & 1.45 & 3.9 \\
       & 0.04 (14)  & 0.219 & 0.089 & 29.81 & 2.79 & 6.9 \\
       & 0.1 (6)    & 0.316 & 0.221 & 65.07 & 6.26 & 89.7 \\
\midrule
Thm. 2 & 0.001 (46) & 0.201 & 0.004 & 23.98 & 1.49 & 14.2 \\
($\alpha=10^{-5}$)  & 0.007 (29) & 0.203 & 0.024 & 24.17 & 1.56 & 14.8 \\
       & 0.04 (14)  & 0.260 & 0.117 & 30.67 & 3.02 & 16.0 \\
       & 0.1 (6)    & 0.450 & 0.383 & 65.83 & 6.73 & 104.0 \\
\midrule
LCLQR  & 0.001 (46) & 0.192 & 0.002 & 25.04 & 1.56 & 4.0 \\
($\alpha=1$)       & 0.007 (29) & 0.192 & 0.011 & 25.26 & 1.61 & 3.9 \\
       & 0.04 (14)  & 0.205 & 0.066 & 32.56 & 2.96 & 5.0 \\
       & 0.1 (6)    & 0.264 & 0.163 & 72.77 & 6.48 & 48.3 \\
\midrule
DDLQR  & 0.001 (46) & 0.192 & 0.002 & 26.89 & 1.66 & 4.0 \\
($\lambda=0.99$)       & 0.007 (29) & 0.192 & 0.011 & 27.13 & 1.73 & 4.2 \\
       & 0.04 (14)  & 0.206 & 0.064 & 35.30 & 3.24 & 5.3 \\
       & 0.1 (6)    & 0.264 & 0.170 & 78.75 & 7.12 & 50.8 \\
\midrule
LQR    & 0.001 (46) & 0.198 & 0.002 & 20.38 & 1.15 & 5.0 \\
(Indirect)& 0.007 (29) & 0.198 & 0.011 & 20.55 & 1.20 & 4.9 \\
       & 0.04 (14)  & 0.212 & 0.065 & 26.22 & 2.28 & 5.9 \\
       & 0.1 (6)    & 0.278 & 0.173 & 58.34 & 5.06 & 91.3 \\
\bottomrule
\end{tabular}
\end{table}

{
\subsection{$(N,\alpha)$ Analysis for Theorems 1 and 2}
Tables~\ref{tab:k1_sweep} and~\ref{tab:k2_sweep} report Monte-Carlo performance varying data length $N$ and step size $\alpha$, confirming the theoretical roles of data richness (Lemma~3) and contraction rate (Lemma~4). {For Theorem~2, the dominant qualitative difference is between the extremely small step size \(\alpha=10^{-5}\) and the larger values \(\alpha=10^{-3},10^{-1},1\). When \(\alpha=10^{-5}\), the contraction factor is very close to \(1\), and the resulting TOA becomes very large, as predicted by Lemma~4. In contrast, once \(\alpha\ge 10^{-3}\), the controller already operates in a sufficiently contractive regime, so the closed-loop transient largely saturates and the state-performance metrics (RMSE and TOA) become similar. In this regime, the remaining differences are more visible in control-effort quantities such as \(E[\|u\|^2]\) and \(\mathrm{RMS}(\Delta u)\) than in final tracking accuracy.}
\paragraph{Role of $\alpha$ (Lemma~4):} 
Lemma~4 provides a contraction bound $\|\mu_k\|_P \le \lambda^{k/2}\|\mu_0\|_P$ with iteration complexity $k = \mathcal{O}((1-\lambda)^{-1}\log(1/\varepsilon))$, where $\lambda \le 1 - c\alpha$ for some $c>0$.  As predicted, increasing $\alpha$ yields a smaller $\lambda$, which significantly reduces the TOA but increases control effort ($\mathbb{E}[\|u\|^2]$ and RMS$(\Delta u)$). For Theorem~1, TOA drops sharply to a minimum transient ($\approx 5$ steps) and saturates. Theorem~2 is acutely sensitive at conservative step sizes (e.g., $\alpha=10^{-5}$); here, $\lambda \approx 1$, driving TOA up to $\mathcal{O}(10^2)$ steps. For Theorem 2, extremely small step sizes (e.g., $\alpha = 10^{-5}$) force the contraction factor to $\lambda \approx 1$. As Lemma 4 predicts, this inflates the iteration bound, matching the observed massive TOA ($\sim 10^2$ steps) despite eventual success.
\paragraph{Role of $N$ (Lemma~3):} 
Lemma~3 guarantees that sufficient samples $N$ well-condition the empirical matrix $\Phi$ (i.e., bounding $\lambda_{\min}(\Phi)$ away from zero). Thus, large $N$ ensures reliable, tuning-insensitive performance. Conversely, small $N$ risks ill-conditioning, making the system highly sensitive—especially when weak contraction (small $\alpha$) amplifies the effect of metric inaccuracies. This can be seen from Tables~\ref{tab:k1_sweep} and \ref{tab:k2_sweep}.
\paragraph{Theorem 1 vs. Theorem 2:} 
While both controllers obey these mechanisms, Theorem~1 demonstrates uniform robustness across the sweep. Consistent with Lemma 4, different constructions yield distinct $\lambda$-to-$\alpha$ constants. For Theorem 2, conservative $\alpha$ yields $\lambda \approx 1$ and massive TOA. Larger $\alpha$ ensures meaningful contraction and smaller TOA, but at a higher actuation cost due to the speed--effort trade-off.
\paragraph{Summary:} 
Overall, the $(N, \alpha)$ sweeps empirically validate the theoretical roles of Lemmas 3 and 4. Specifically, $N$ dictates the conditioning and reliability of the learned metric (Lemma 4), whereas $\alpha$ controls the contraction rate and transient speed (Lemma 3). This highlights a clear, observable trade-off: faster arrival times (smaller TOA) inherently require greater control effort (higher energy and RMS($\Delta u$)).
}

\begin{table}[htbp]
\centering
\caption{{Theorem 1 performance ($T=300, N_m=500$). }}
\label{tab:k1_sweep}
\setlength{\tabcolsep}{3pt}
\small % Slightly smaller font to guarantee fit
\begin{tabular}{cc | ccccc}
\toprule
$N$ & $\alpha$ & RMSE & Final Err & $\mathbb{E}[\|u\|^2]$ & RMS & TOA \\
 & & (pos) & ($\times 10^{-4}$) & & ($\Delta u$) & (mean) \\
\midrule
16 & 0.01 & 0.255 & 1.65 & 49.70 & 2.63 & 15.9 \\
   & 0.1  & 0.249 & 1.44 & 58.97 & 4.02 & 10.0 \\
   & 1    & 0.244 & 1.27 & 56.04 & 4.56 & 5.0  \\
   & 10   & 0.244 & 1.27 & 52.68 & 4.12 & 5.0  \\
\hline
32 & 0.01 & 0.348 & 2.47 & 18.57 & 0.95 & 26.0 \\
   & 0.1  & 0.246 & 1.50 & 44.18 & 2.37 & 6.0  \\
   & 1    & 0.246 & 1.31 & 45.85 & 2.70 & 5.0  \\
   & 10   & 0.244 & 1.31 & 49.76 & 3.24 & 5.0  \\
\hline
64 & 0.01 & 0.347 & 2.79 & 16.60 & 0.87 & 23.0 \\
   & 0.1  & 0.244 & 1.68 & 46.62 & 2.54 & 7.0  \\
   & 1    & 0.244 & 1.28 & 55.50 & 3.72 & 5.0  \\
   & 10   & 0.244 & 1.29 & 55.68 & 3.73 & 5.0  \\
\hline
128& 0.01 & 0.384 & 2.69 & 12.29 & 0.40 & 29.0 \\
   & 0.1  & 0.246 & 1.38 & 44.20 & 2.25 & 6.0  \\
   & 1    & 0.243 & 1.24 & 53.96 & 3.98 & 5.0  \\
   & 10   & 0.243 & 1.28 & 54.34 & 3.82 & 5.0  \\
\bottomrule
\end{tabular}
\end{table}

\begin{table}[htbp]
\centering
\caption{{Theorem 2 performance ($T=300, N_m=500$).}}
\label{tab:k2_sweep}
\setlength{\tabcolsep}{3pt}
\small
\begin{tabular}{cc | ccccc}
\toprule
$N$ & $\alpha$ & RMSE & Final Err & $\mathbb{E}[\|u\|^2]$ & RMS & TOA \\
 & & (pos) & ($\times 10^{-4}$) & & ($\Delta u$) & (mean) \\
\midrule
16 & $10^{-5}$ & 0.286 & 3.11 & 48.46 & 2.55 & 64.0 \\
   & $10^{-3}$ & 0.268 & 4.10 & 181.8 & 20.2 & 27.0 \\
   & $10^{-1}$ & 0.267 & 3.39 & 113.9 & 10.9 & 16.6 \\
   & $10^{0}$  & 0.266 & 3.14 & 127.3 & 13.4 & 18.0 \\
\hline
32 & $10^{-5}$ & 0.653 & 40.8 & 7.054 & 0.33 & 144.0 \\
   & $10^{-3}$ & 0.248 & 1.97 & 64.94 & 5.67 & 6.0 \\
   & $10^{-1}$ & 0.249 & 2.37 & 63.60 & 5.51 & 7.0 \\
   & $10^{0}$  & 0.248 & 2.19 & 66.10 & 5.68 & 6.0 \\
\hline
64 & $10^{-5}$ & 0.952 & 278.0 & 3.143 & 0.29 & 216.8 \\
   & $10^{-3}$ & 0.262 & 1.82 & 47.24 & 4.62 & 7.0 \\
   & $10^{-1}$ & 0.248 & 2.76 & 47.67 & 3.81 & 5.0 \\
   & $10^{0}$  & 0.247 & 3.19 & 49.54 & 4.36 & 5.0 \\
\hline
128& $10^{-5}$ & 0.606 & 102.0 & 6.607 & 0.17 & 145.7 \\
   & $10^{-3}$ & 0.256 & 1.63 & 47.76 & 3.22 & 6.0 \\
   & $10^{-1}$ & 0.252 & 2.41 & 55.54 & 5.74 & 8.0 \\
   & $10^{0}$  & 0.252 & 2.15 & 55.81 & 6.07 & 8.0 \\
\bottomrule
\end{tabular}
\end{table}

\section{Related Work}
This work connects data-driven control, natural gradient optimization, and robustness in policy design.

\textbf{Gradient-Based and Natural Gradient Control.} NGD methods~\cite{6790500,martens2020new} excel in curved parameter spaces. While widespread in RL, NGD remains underutilized in control; existing gradient-based trajectory shaping~\cite{esmzad2024gdl,esmzad2026natural} often assumes full model knowledge, neglecting data-driven uncertainty.

\textbf{Direct Data-Driven Control.} Model-free synthesis from trajectories~\cite{DEPERSIS2021109548,dorfler2022bridging,zhao2024data} (e.g., DDLQR~\cite{ESMZAD2025112197} or data-driven Riccati equations~\cite{rantzer2024data,rantzer2024lqdc}) avoids system identification. However, these often rely on unintuitive $Q, R$ tuning. Our framework integrates the FIM for preconditioned updates aligned with the data's uncertainty structure.

\textbf{Covariance-Aware Robust Control.} Traditional covariance control~\cite{hotz1987covariance} requires known matrices or noise statistics. Recent data-driven efforts~\cite{zhao2024data} use sample covariances but frequently omit higher-order noise propagation. We explicitly incorporate these second-order terms to bolster robustness in stochastic, low-data regimes.

\textbf{Optimization-Inspired Feedback Control.} Recasting optimization as dynamical systems~\cite{HAUSWIRTH2024100941,lessard2016analysis,padmanabhan2024analysis} allows for analysis via IQCs or passivity~\cite{nayyer2022passivity}. Our method adopts this philosophy by embedding NGD within the closed-loop, enabling geometrically aware, data-driven trajectory shaping.

\section{Conclusion and Future Work} \label{sec:conc}
This paper introduces an NGD-based data-driven control methodology offering a geometric perspective on trajectory shaping. By preconditioning policy updates with the inverse of the closed-loop covariance matrix, the framework leverages the FIM to handle uncertainty. Our contributions include formal LMI-based stability guarantees and experimental validation on a physical Mecanum-wheeled robot.

Despite these results, several limitations remain. The current LTI formulation with Gaussian noise should be extended to nonlinear systems via linearization or kernel embeddings and adapted for non-stationary environments. While Theorem 2 facilitates online updates, the framework currently operates in batch mode, and the underlying SDP synthesis may face scaling issues with ill-conditioned or high-dimensional matrices. Additionally, while the scalar $\alpha$ provides interpretability, its tuning sensitivity—particularly in the covariance-based formulation—requires care. Finally, future work will incorporate state and input constraints through MPC-like mechanisms and expand testing to broader tasks like trajectory tracking and disturbance rejection across diverse robotic platforms.

\section*{APPENDIX}
This appendix provides additional technical details and experimental results to support the main claims of the paper.

% \begin{wrapfigure}{r}{0.5\columnwidth}
\begin{figure}[htbp]
    \centering
    \begin{minipage}{0.38\linewidth}
        \centering
        \includegraphics[width=\linewidth,trim=0cm 1cm 0cm 1cm, clip]{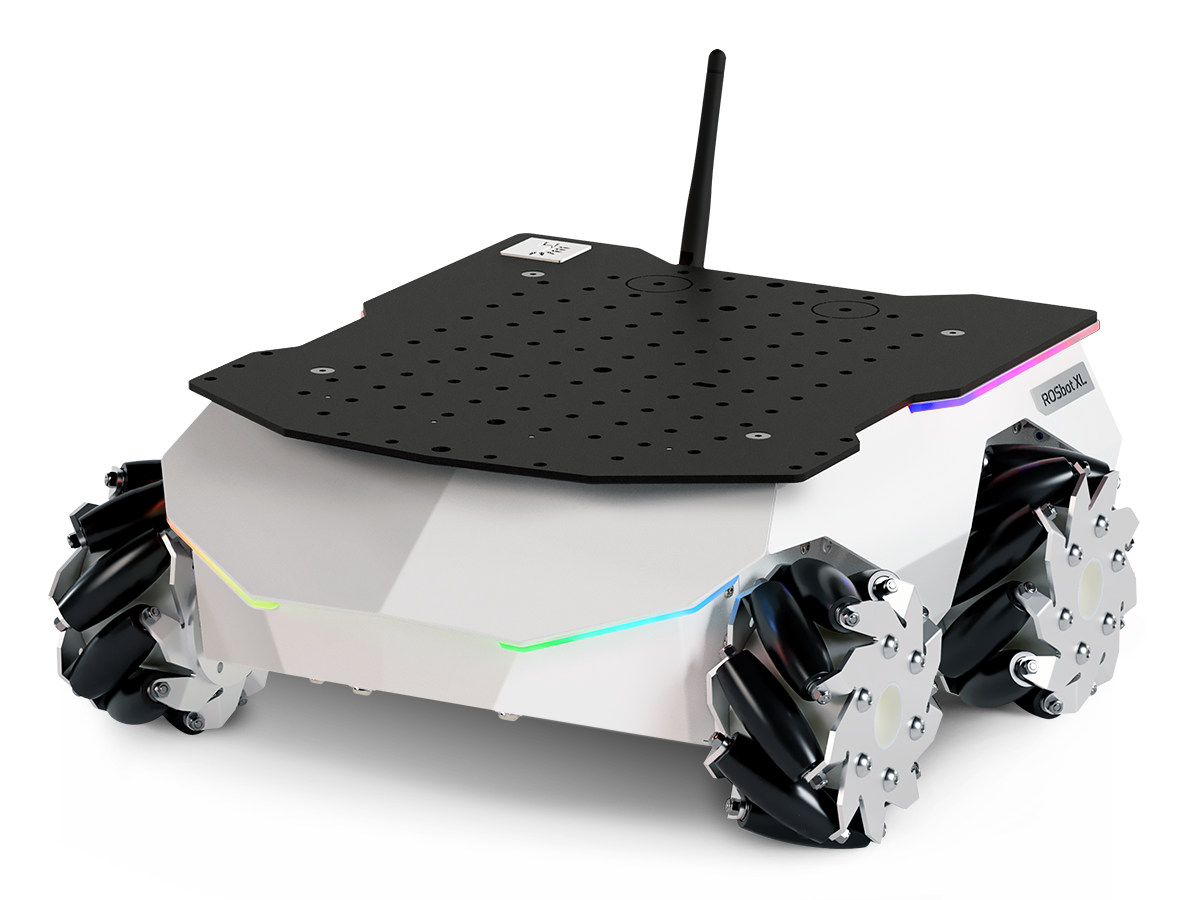}
    \end{minipage}
    \hfill
    \begin{minipage}{0.58\linewidth}
        \centering
        \includegraphics[width=\linewidth, trim=0cm 1cm 0cm 1cm, clip]{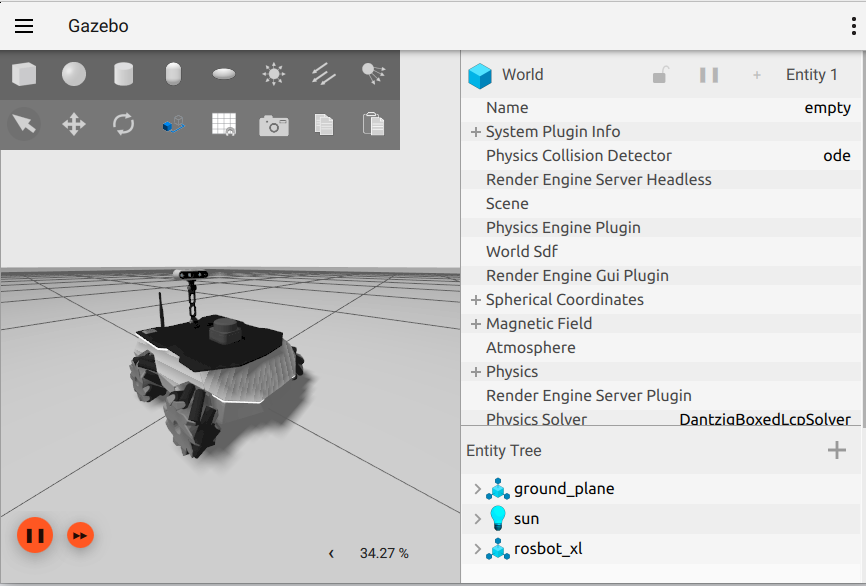}
    \end{minipage}
    \caption{The physical (left) and simulated ROSbot XL platform (right) used in our simulations and experiments.}
    \label{fig:rosbot_real_vs_gazebo}
    % \vspace{-1.cm}
\end{figure}
% \end{wrapfigure}

\subsection{Robot Model}\label{app:robot}

We consider a four-wheeled omnidirectional mobile robot with Mecanum wheels. Its planar nonlinear kinematics are given by~\cite{10908211,SUN2021107128,1087767}
\begin{equation}
\begin{aligned}
\dot{x}_{\mathrm{global}} &= \cos(\phi)\, v_x - \sin(\phi)\, v_y,\\
\dot{y}_{\mathrm{global}} &= \sin(\phi)\, v_x + \cos(\phi)\, v_y,\\
\dot{\phi} &= \omega,
\end{aligned}
\label{eq:nonlinear_model}
\end{equation}
where \(v_x\) and \(v_y\) are body-frame linear velocities, \(\phi\) is the heading angle, and \(\omega\) is the angular velocity. These body-frame velocities are generated by the wheel speeds \(\omega_i\), \(i=1,2,3,4\), as
\begin{equation}
\begin{bmatrix} v_x \\ v_y \\ \omega \end{bmatrix}
=
\frac{r}{4}
\begin{bmatrix}
1 & 1 & 1 & 1 \\
-1 & 1 & 1 & -1 \\
-\frac{1}{l+w} & \frac{1}{l+w} & -\frac{1}{l+w} & \frac{1}{l+w}
\end{bmatrix}
\begin{bmatrix}
\omega_1 \\ \omega_2 \\ \omega_3 \\ \omega_4
\end{bmatrix},
\label{eq:mecanum_map}
\end{equation}
where \(r=0.05\,\mathrm{m}\) is the wheel radius, and \(l=0.83\,\mathrm{m}\), \(w=0.163\,\mathrm{m}\) are half the robot length and width.

Although the proposed framework does not require an explicit model, for evaluation we linearize \eqref{eq:nonlinear_model} around \(\phi\approx 0\) and discretize it with sampling time \(\Delta t=0.5\,\mathrm{s}\), yielding
\begin{equation}
x_{k+1}=Ax_k+Bu_k+w_k,
\qquad
x_k\in\mathbb R^3,\;\; u_k\in\mathbb R^4,
\label{eq:linear_model}
\end{equation}
where \(x_k=[x,y,\phi]^\top\), \(u_k=[\omega_1,\omega_2,\omega_3,\omega_4]^\top\), and \(w_k\) is process noise. Under this approximation,
\begin{equation}
A=I_3,
\qquad
B=\Delta t\,\frac{r}{4}
\begin{bmatrix}
1 & 1 & 1 & 1 \\
-1 & 1 & 1 & -1 \\
-\frac{1}{l+w} & \frac{1}{l+w} & -\frac{1}{l+w} & \frac{1}{l+w}
\end{bmatrix}.
\label{eq:AB_robot}
\end{equation}

\subsection{Data Collection Details}
\label{App:data_collection}
We used predefined wheel-velocity patterns with translational and mild rotational motions, selecting one at random at each time step to ensure persistent excitation.
\begin{table}[htbp]
\centering
\caption{Predefined Wheel Velocity Patterns and Associated Motions}
\begin{tabular}{@{}c l l@{}}
\toprule
\textbf{Index} & \textbf{Wheel Speeds $[\omega_1, \omega_2, \omega_3, \omega_4]$} & \textbf{Motion Type} \\
\midrule
1 & $[1, 1, 1, 1]$ & Forward \\
2 & $[-1, -1, -1, -1]$ & Backward \\
3 & $[-1, 1, 1, -1]$ & Left \\
4 & $[1, -1, -1, 1]$ & Right \\
5 & $[0, 1, 1, 0]$ & Forward-Left \\
6 & $[1, 0, 0, 1]$ & Forward-Right \\
7 & $[-1, 0, 0, -1]$ & Backward-Left \\
8 & $[0, -1, -1, 0]$ & Backward-Right \\
9 & $[-0.5, 0.5, -0.5, 0.5]$ & Spin Clockwise \\
10 & $[0.5, -0.5, 0.5, -0.5]$ & Spin Counter-Clockwise \\
\bottomrule
\end{tabular}
\label{tab:motion_patterns}
\end{table}

These patterns provide rich excitation while preserving safety and reliable odometry; the resulting trajectories and inputs are shown in Figs.~\ref{fig:data_collection_states} and~\ref{fig:data_collection_inputs}.
Since we use odometry for state estimation, all experiments and simulations start from the fixed origin \(x_0=\begin{bmatrix}0 & 0 & 0\end{bmatrix}^\top\), corresponding to zero position and orientation.

\begin{figure}[h]
    \centering
    \includegraphics[width=0.9\linewidth]{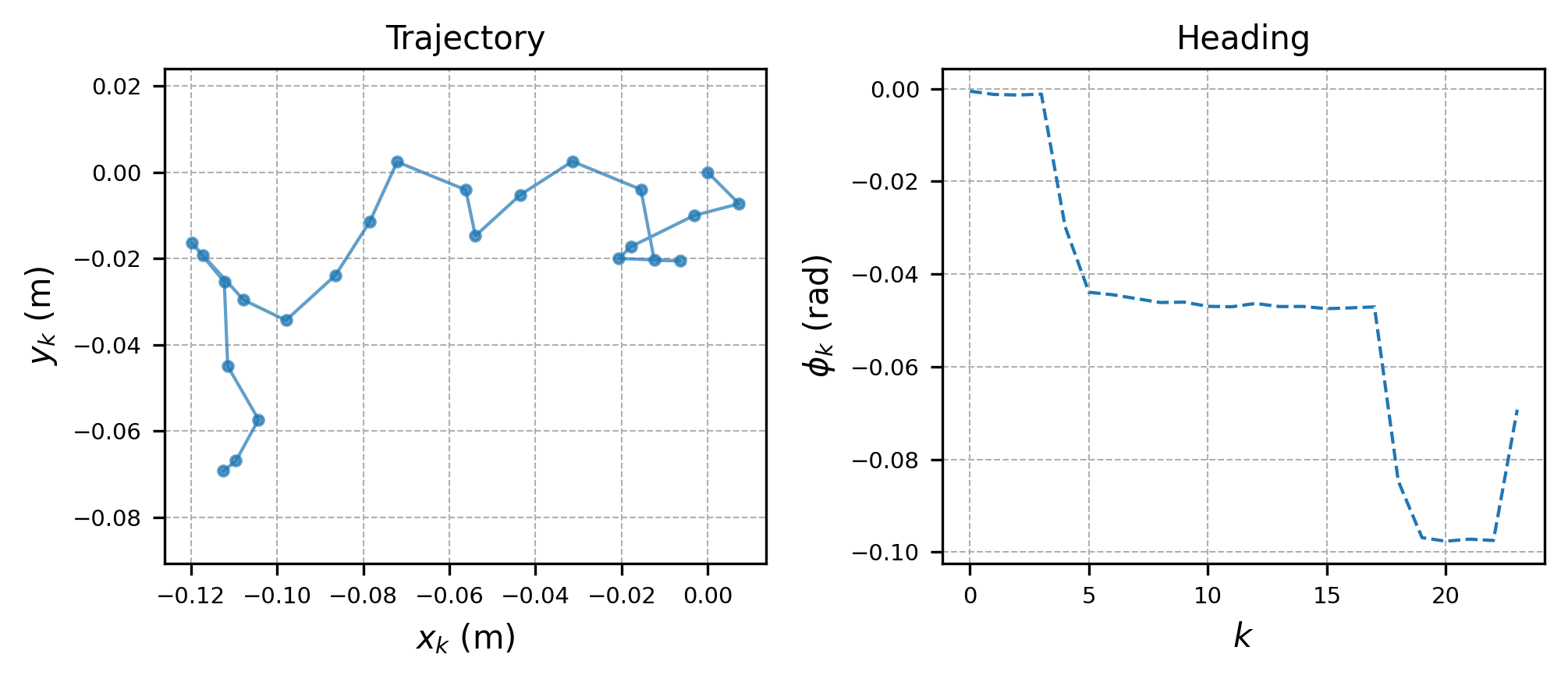}
        \caption{Data collected from the Gazebo simulation. Left: robot trajectory in the $(x_k, y_k)$ plane during the data collection phase. Right: the corresponding heading angle $\phi_k$ over time.}

    \label{fig:data_collection_states}
\end{figure}

\begin{figure}[h]
    \centering
    \includegraphics[width=0.9\linewidth]{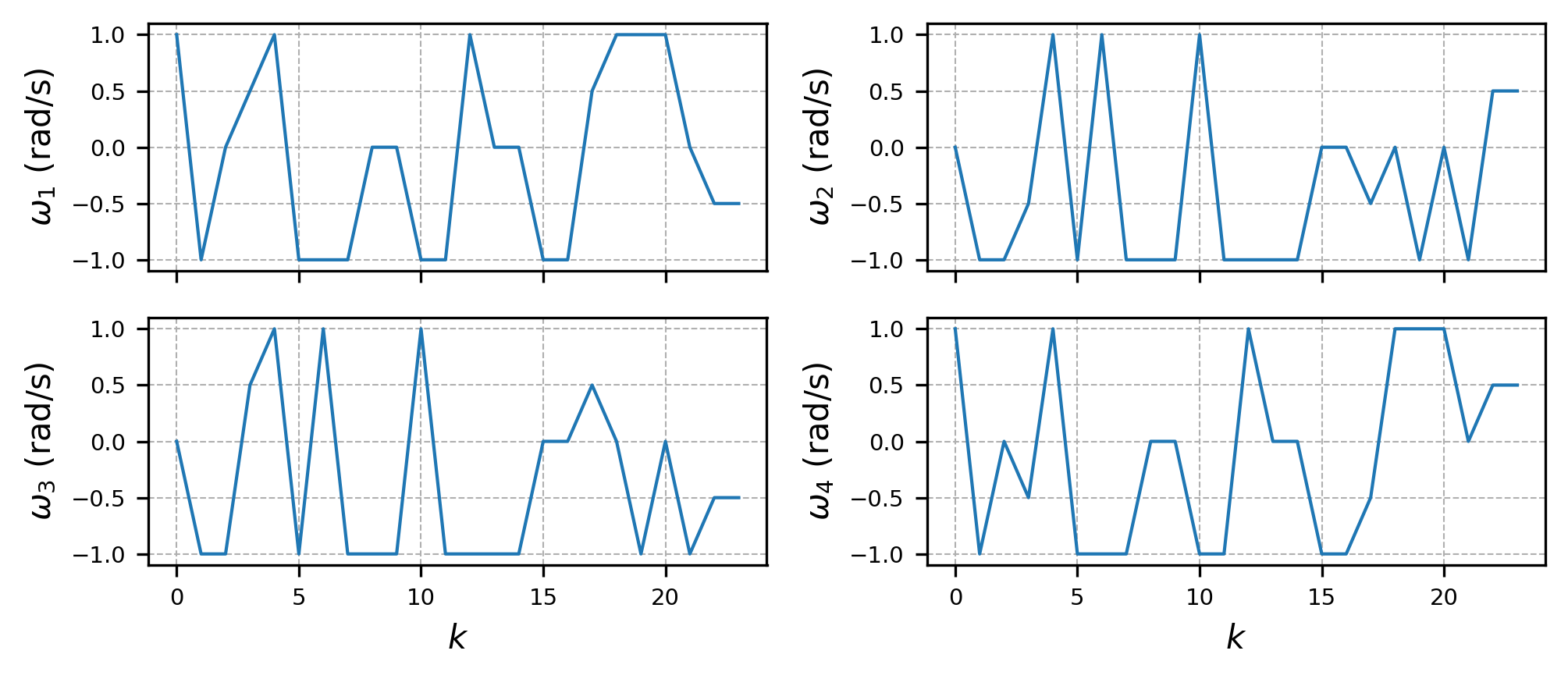}
     \caption{Control inputs (wheel speeds) $\omega_1$ to $\omega_4$ applied to the robot in Gazebo during data collection.  The inputs were selected from structured translation-only (with slight rotational) patterns to preserve the validity of the linear model.}
    \label{fig:data_collection_inputs}
\end{figure}

\subsection{Statistical Metrics for Monte Carlo Analysis} \label{app:monte}

To evaluate the consistency and robustness of the NGD controller (Theorem 1), we compute the mean and standard deviation of the system trajectories and control inputs across multiple Monte Carlo trials. These statistics help quantify how much the behavior of the system varies under stochastic disturbances.

Let \( x_k^{(i)} \) denote the value of a state or control variable at time step \( k \) during the \( i^\text{th} \) trial, for \( i = 1, \dots, N_m \). The following metrics are computed:

\begin{itemize}
    \item \textbf{Mean:} The sample mean at time step \( k \) is given by
    \[
    \bar{x}_k = \frac{1}{N_m} \sum_{i=1}^{N_m} x_k^{(i)}.
    \]
    
    \item \textbf{Standard deviation:} The sample standard deviation measures the variability across trials and is computed as
    \[
    \sigma_k = \sqrt{\frac{1}{N_m - 1} \sum_{i=1}^{N_m} \left( x_k^{(i)} - \bar{x}_k \right)^2}.
    \]
    
    \item \textbf{Visualization:} In the trajectory and control input plots in figures ~\ref{fig:mc_xy_phi_small}, ~\ref{fig:mc_controls_small},~\ref{fig:mc_xy_phi_large}, and ~\ref{fig:mc_controls_large}, we display the shaded region corresponding to
    $\bar{x}_k \pm \sigma_k,$
    which captures the interval likely to contain approximately 68\% of the data if the distribution is Gaussian. This band provides a visual summary of variability in system behavior.
\end{itemize}

Final values such as \( x_T \), \( y_T \), and \( \phi_T \) are also reported with their sample mean and standard deviation to compare the impact of different step-size values \( \alpha \), as summarized in Table~\ref{tab:mc_alpha_comparison}.

\begin{figure}[h]
    \centering
    \includegraphics[width=0.9\linewidth]{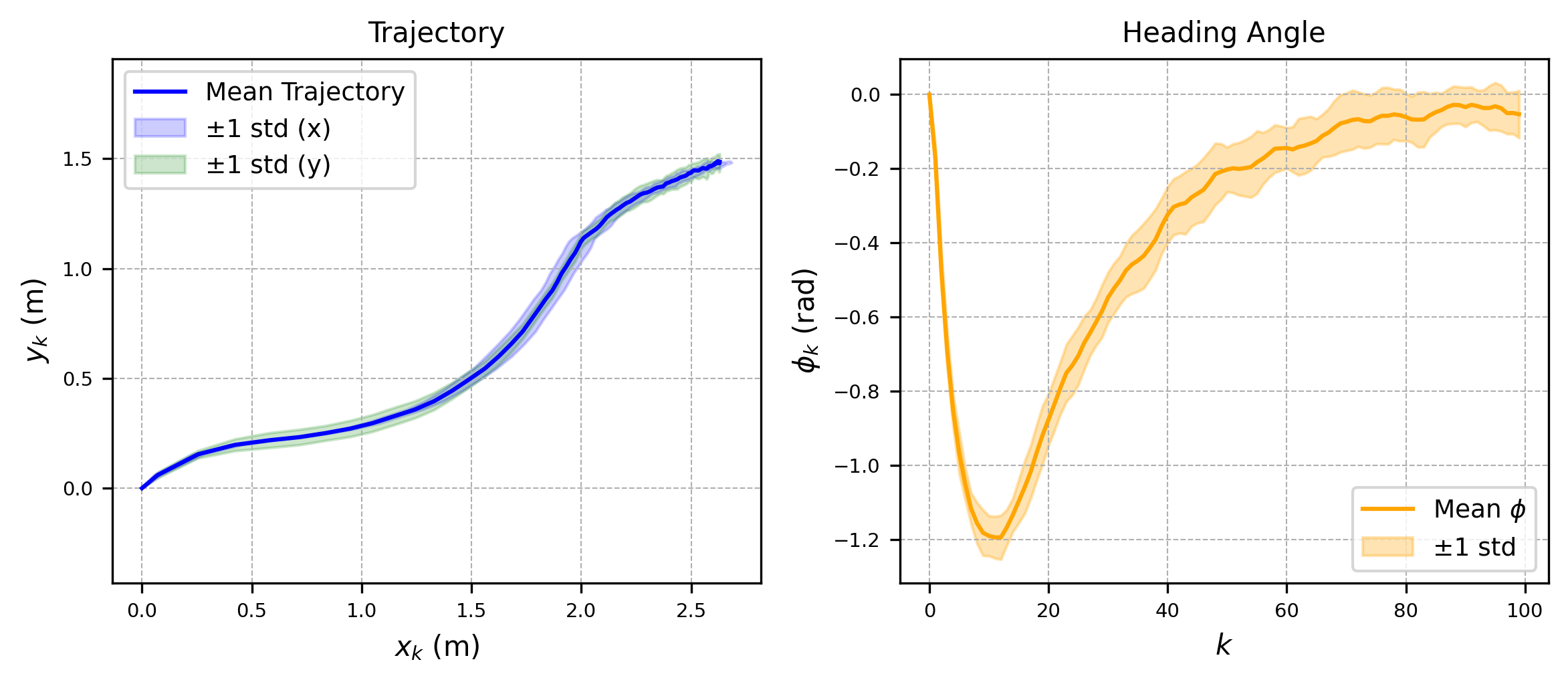}
        \caption{
    Monte Carlo $x$-$y$ trajectory and heading for $\alpha = 10^{-5}$.
    }
    \label{fig:mc_xy_phi_small}
\end{figure}

\begin{figure}[h]
    \centering
    \includegraphics[width=0.5\textwidth]{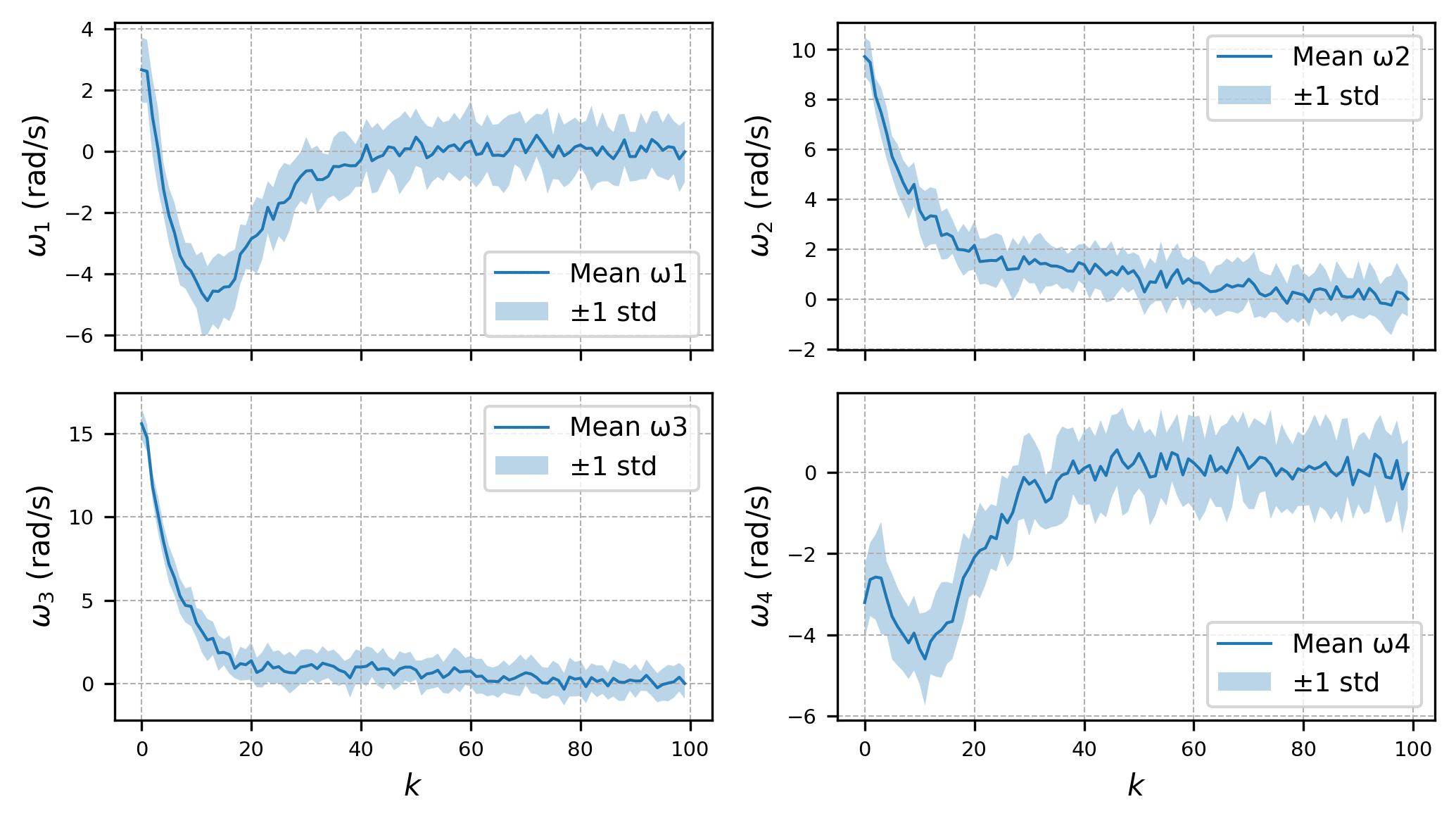}
    \caption{
    Control inputs $\omega_1$–$\omega_4$ for $\alpha = 10^{-5}$.
    }
    \label{fig:mc_controls_small}
\end{figure}

\begin{figure}[h]
    \centering
    \includegraphics[width=0.9\linewidth]{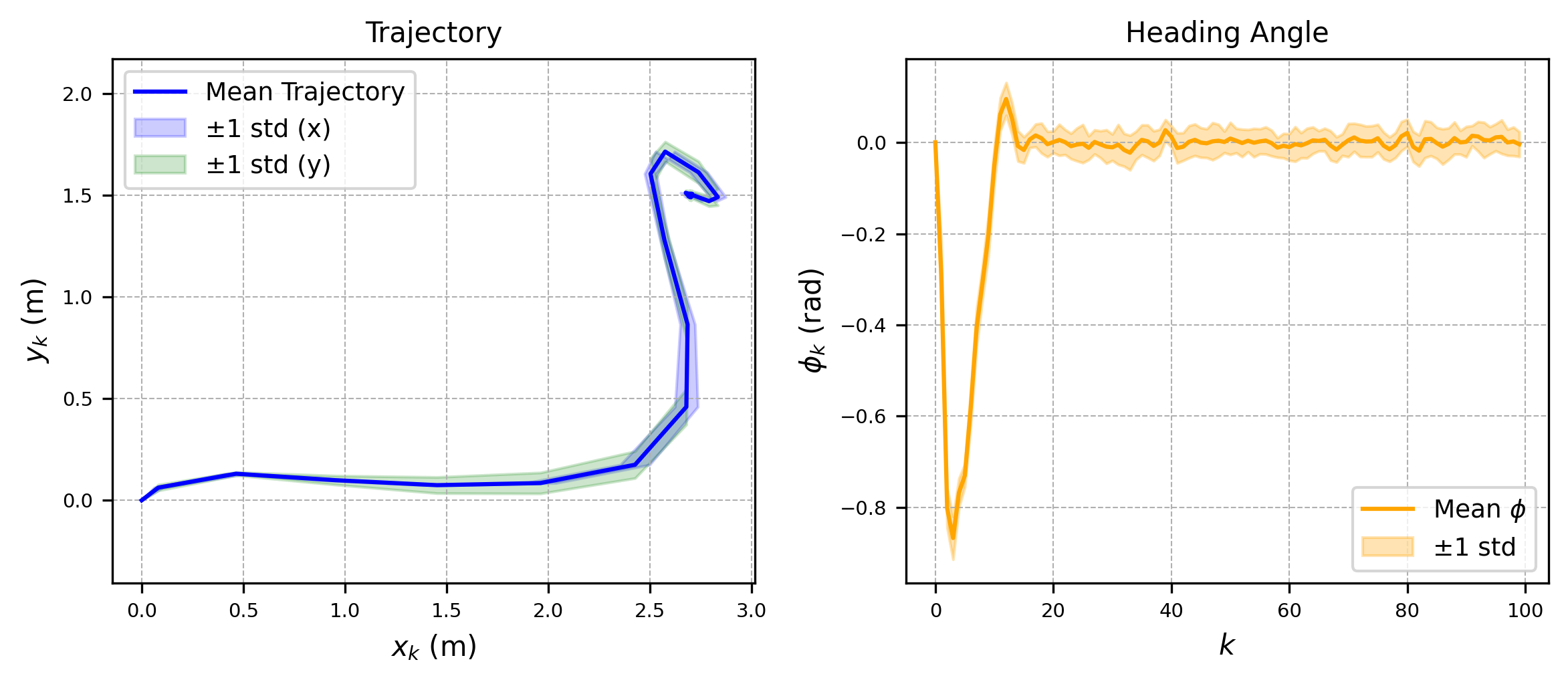}
        \caption{
    Monte Carlo $x$-$y$ trajectory and heading for $\alpha = 0.2$.
    }
    \label{fig:mc_xy_phi_large}
\end{figure}

\begin{figure}[h]
    \centering
    \includegraphics[width=0.5\textwidth]{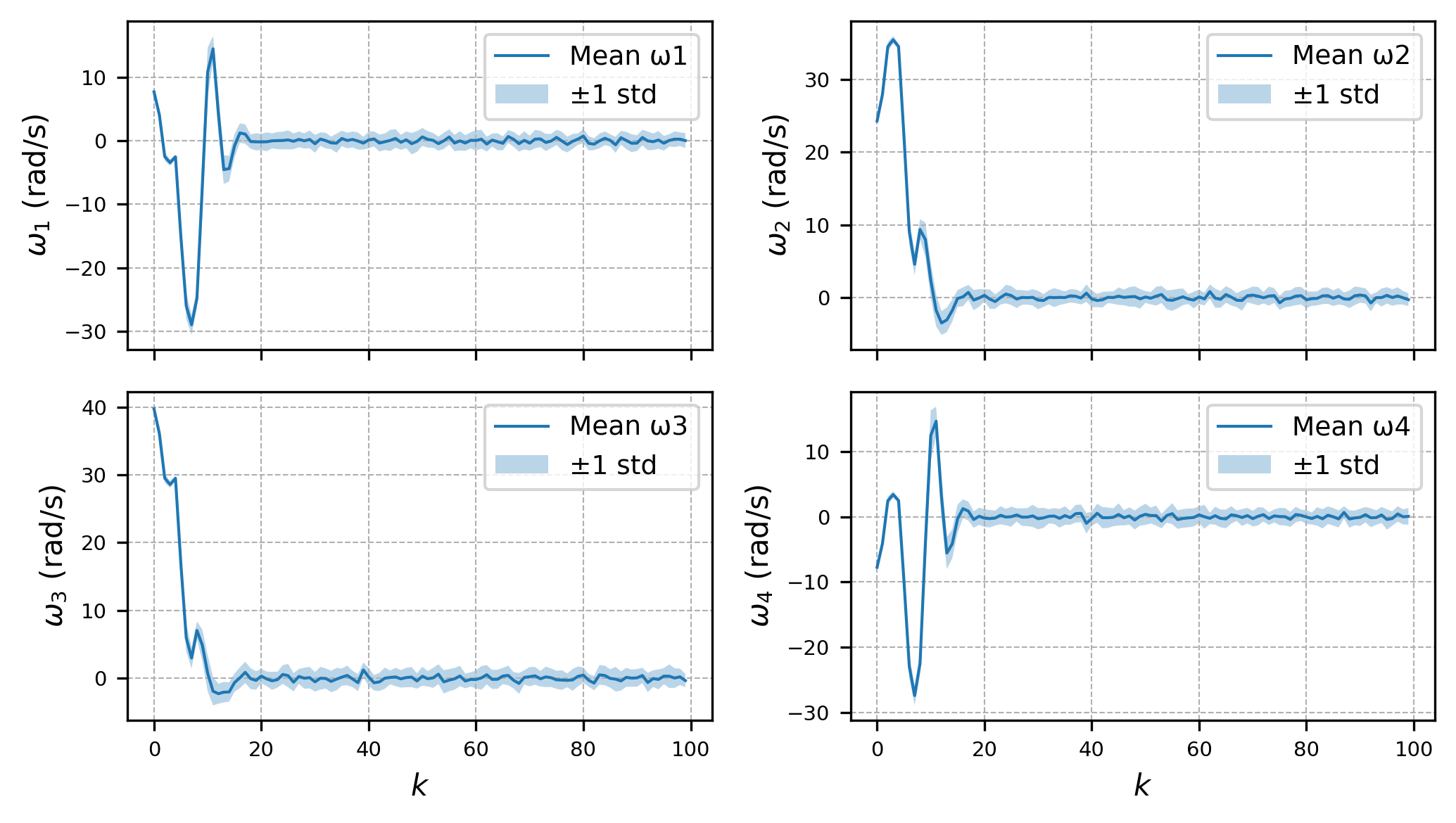}
    \caption{
    Control inputs $\omega_1$–$\omega_4$ for $\alpha = 0.2$.
    }
    \label{fig:mc_controls_large}
\end{figure}

\subsection{Additional Simulation Results from Subsection V.\ref{subsec:behavior_sim}} \label{app:sim}
This section provides supplementary simulation plots supporting subsection V.\ref{subsec:behavior_sim}, where we compare the effect of varying the step size $\alpha$ in Theorems 1 and 2, and explore DDLQR performance under multiple $(Q, R)$ configurations reported in Table 3.

The trajectories of the robot and inputs for various $\alpha$ using Theorem~1 are shown in figures~\ref{fig:dd1_states} and~\ref{fig:dd1_controls}. Similar results using Theorem 2 are shown in figures \ref{fig:dd2_states} and~\ref{fig:dd2_controls}. The trajectories of the robot and inputs using DDLQR approach for the distinct $(Q,R)$ pairs in Table 3 are shown in figures~\ref{fig:ddlqr_states} and~\ref{fig:ddlqr_controls}.

\begin{figure}[h]
    \centering
    \includegraphics[width=1\linewidth]{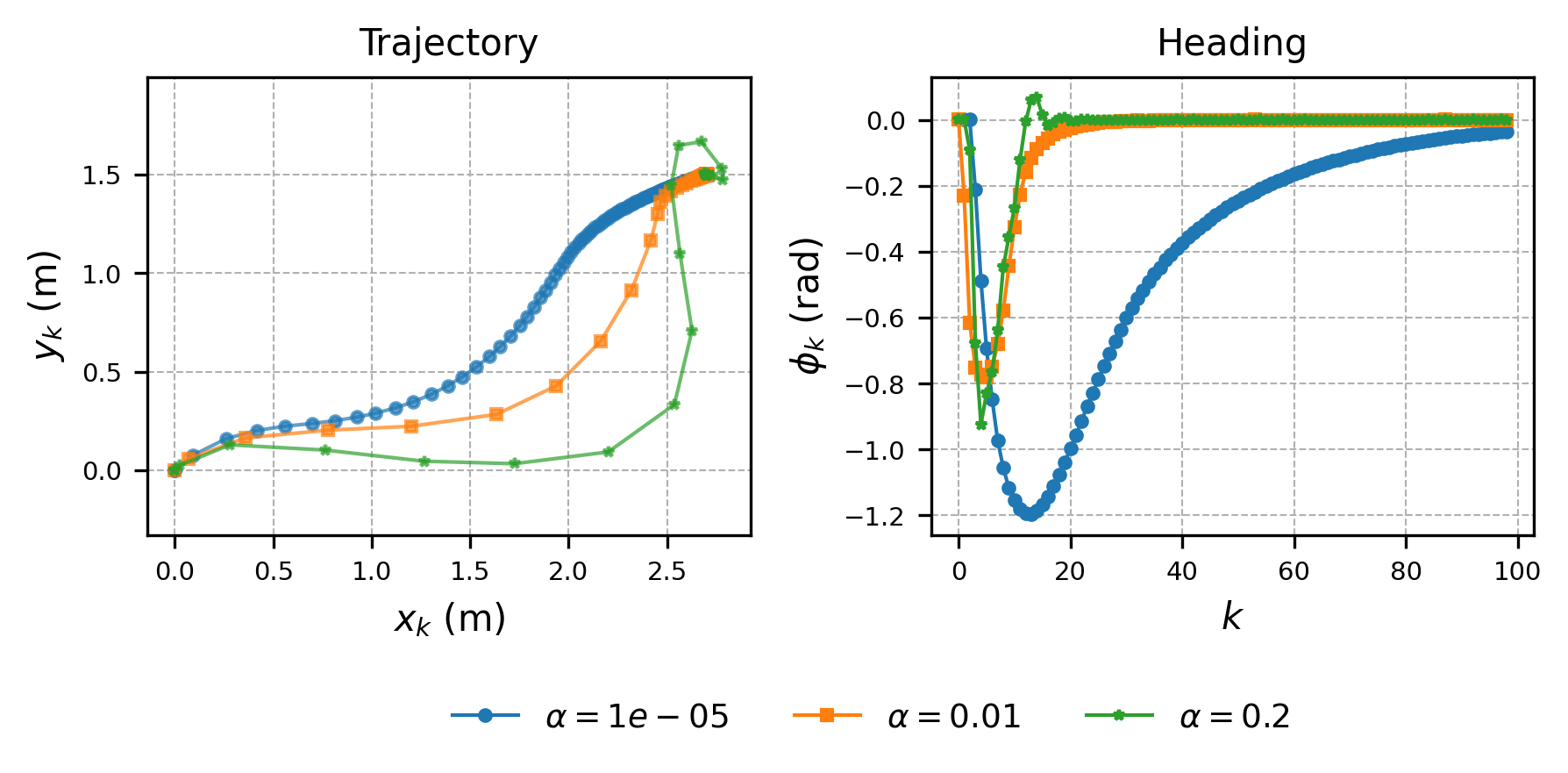}
        \caption{Robot trajectory obtained for various $\alpha$ in Theorem 1}
    \label{fig:dd1_states}
\end{figure}

\begin{figure}[h]
    \centering
    \includegraphics[width=1\linewidth]{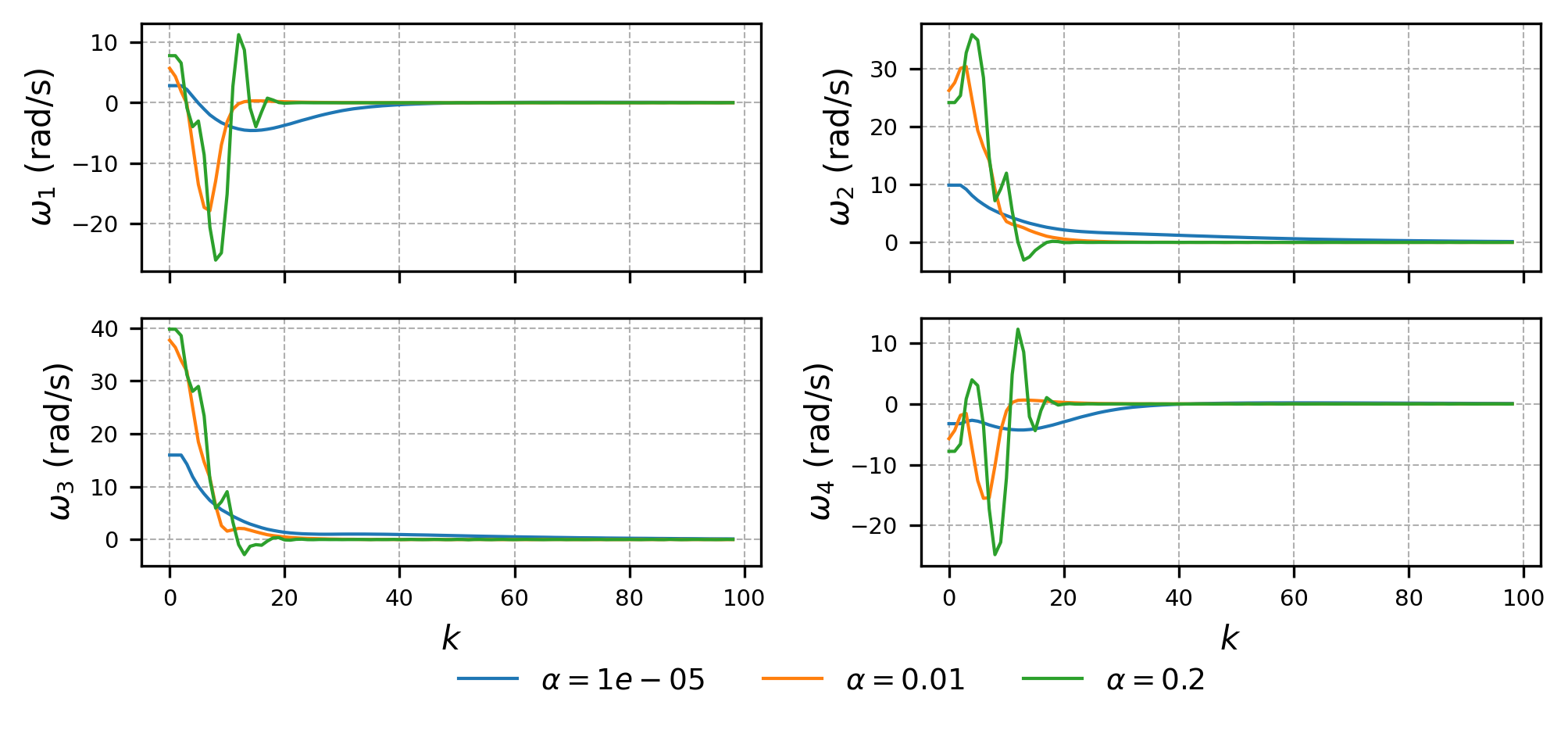}
        \caption{Wheel speeds obtained for various $\alpha$ in Theorem 1}
    \label{fig:dd1_controls}
\end{figure}

\begin{figure}[h]
    \centering
    \includegraphics[width=1\linewidth]{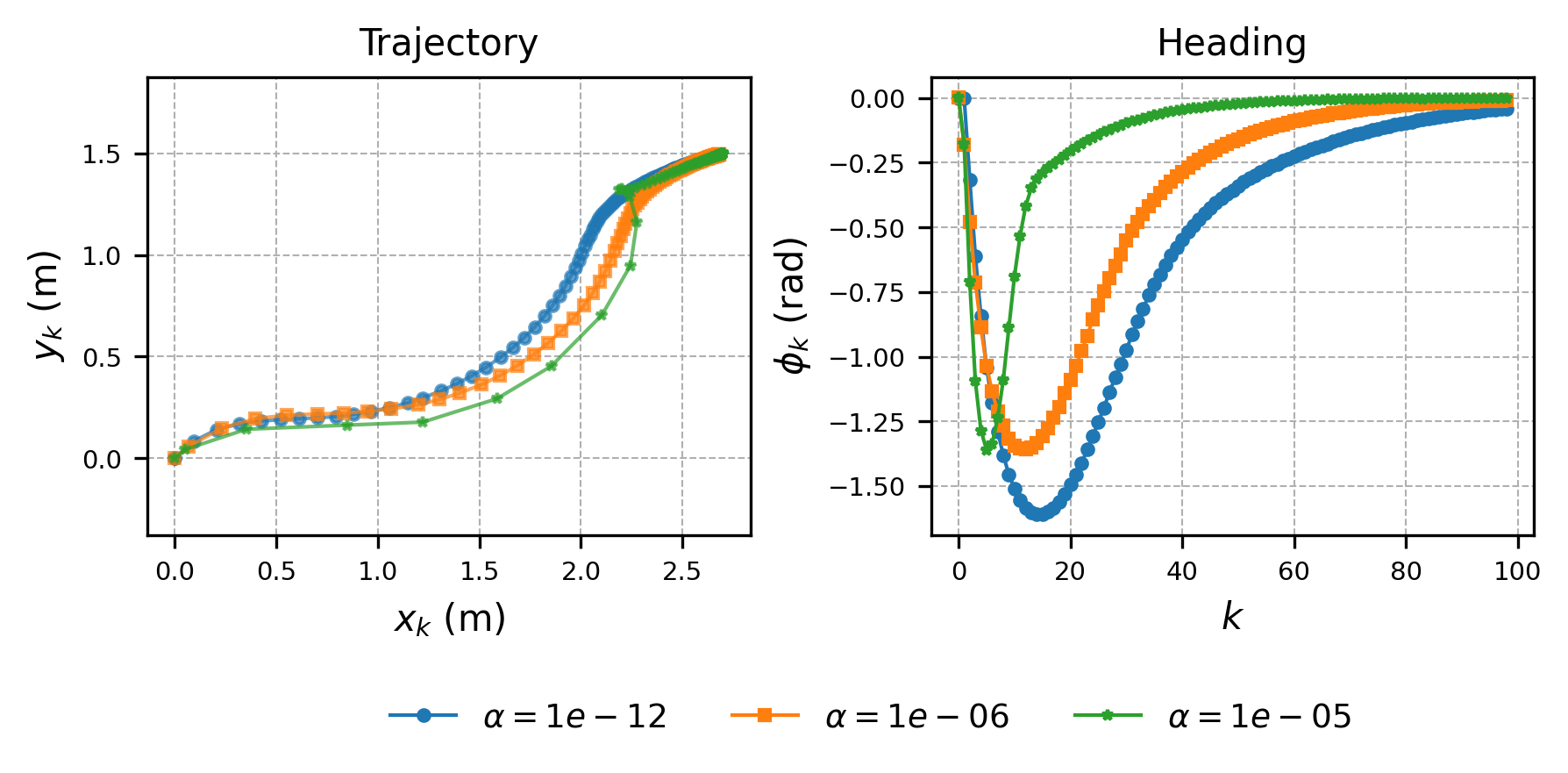}
        \caption{Robot trajectory obtained for various $\alpha$ in Theorem 2}
    \label{fig:dd2_states}
\end{figure}

\begin{figure}[h]
    \centering
    \includegraphics[width=1\linewidth]{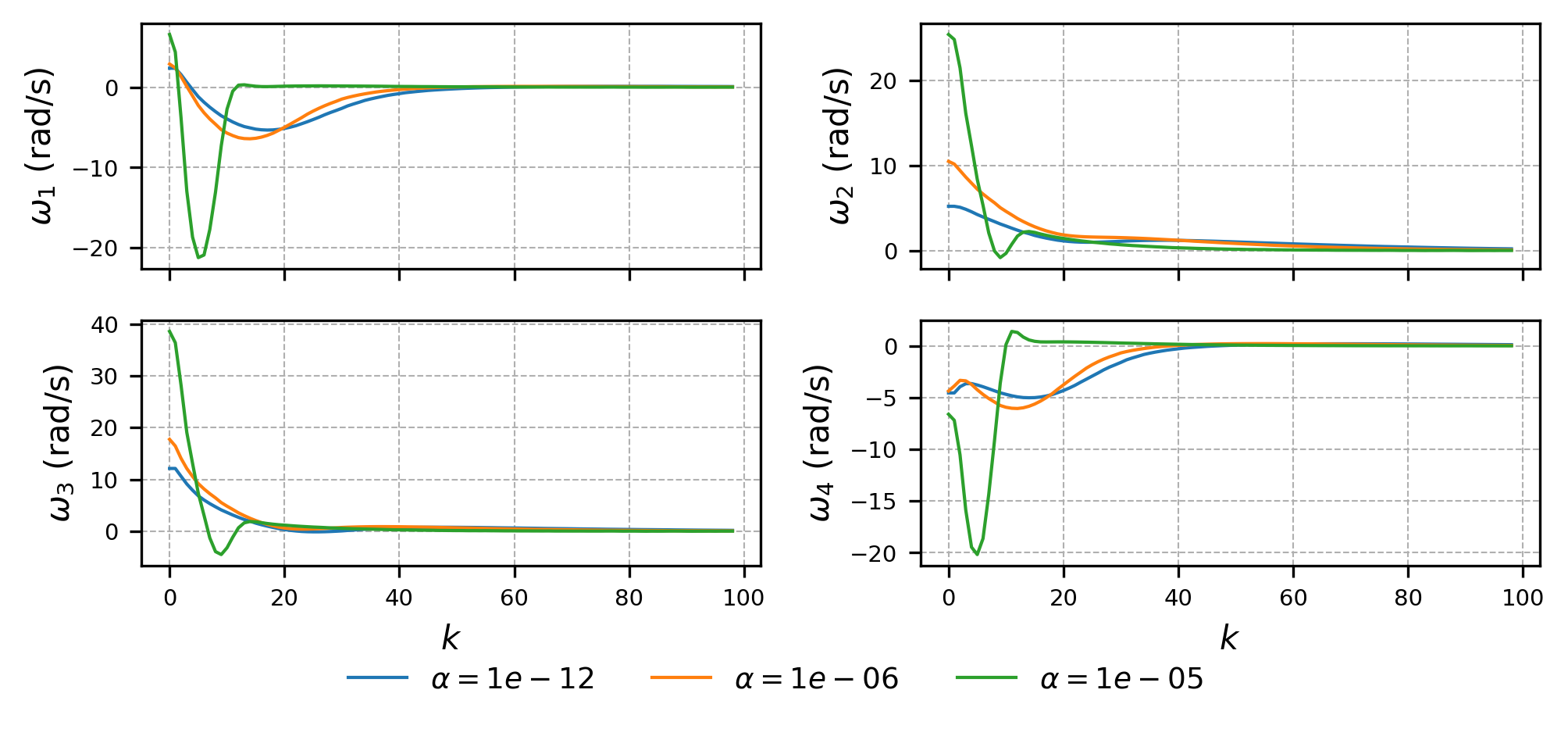}
        \caption{Wheel speeds obtained for various $\alpha$ in Theorem 2}
    \label{fig:dd2_controls}
\end{figure}

\begin{figure}[h]
    \centering
    \includegraphics[width=1\linewidth]{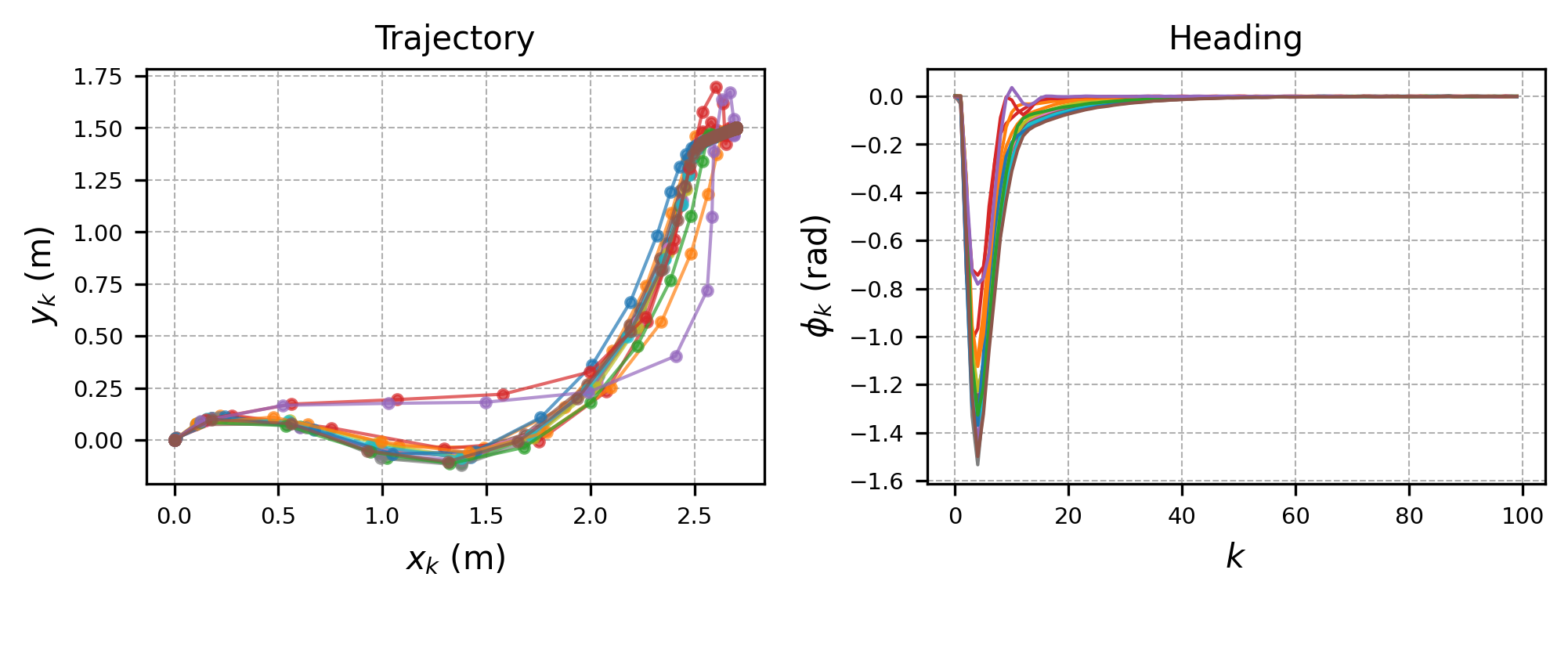}
        \caption{Robot trajectory obtained for various $Q$ and $R$ in DDLQR}
    \label{fig:ddlqr_states}
\end{figure}

\begin{figure}[h]
    \centering
    \includegraphics[width=1\linewidth]{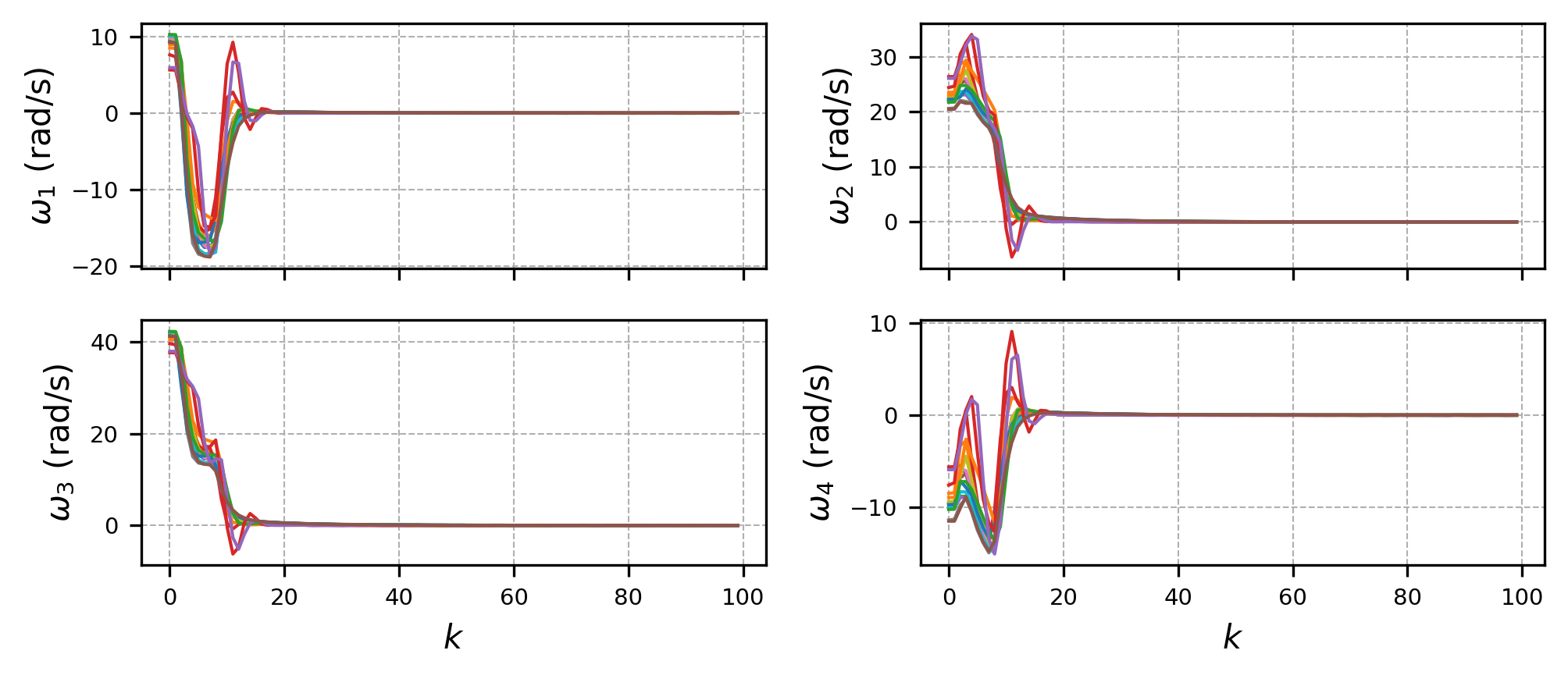}
        \caption{Wheel speeds obtained for various $Q$ and $R$ in DDLQR}
    \label{fig:ddlqr_controls}
\end{figure}

\begin{table}[htbp]
\centering
\caption{Selected DDLQR tuning sets. $Q$ and $R$ are diagonal matrices (only diagonal elements are listed). }
\label{tab:qr-tuning-sets}
% Using p{} columns to force text wrapping and strictly limit table width
\begin{tabular}{@{} c p{2.4cm} p{1.6cm} p{2.6cm} @{}}
\toprule
\textbf{Set} & \boldmath$Q$ \textbf{Diag.} & \boldmath$R$ \textbf{Diag.} & \textbf{Interpretation}\\
\midrule
1  & $[1, 1, 10]$ & $0.1 I_4$ & Balanced, light input. \\
2  & $[1, 1, 10]$ & $0.5 I_4$ & Smoother than Set 1. \\
3  & $[0.5, 0.5, 50]$ & $0.1 I_4$ & Strong heading focus. \\
4  & $[10, 10, 50]$ & $0.05 I_4$ & Aggressive tracking. \\
5  & $[5, 5, 10]$ & $1 I_4$ & Balanced penalties. \\
6  & $[5, 5, 2]$ & $0.1 I_4$ & Position over heading. \\
7  & $[1, 1, 5]$ & $5 I_4$ & Conservative control. \\
8  & $[20, 20, 10^{-5}]$ & $0.2 I_4$ & Precise $(x, y)$ only. \\
9  & $[20, 20, 30]$ & $2 I_4$ & Strong overall tracking. \\
10 & $[1, 10, 5]$ & $[0.1, 0.1,$ \newline $0.5, 0.5]$ & High $y$ cost, smooth rear. \\
11 & $[2, 2, 30]$ & $[0.05, 0.2,$ \newline $0.2, 0.8]$ & Heading focus, limits W4. \\
12 & $[5, 1, 10]$ & $[0.3, 0.4,$ \newline $0.2, 0.5]$ & High $x$/heading cost. \\
13 & $[1, 2{\times}10^4, 0.1]$ & $10^{-3} I_4$ & Max $y$ precision. \\
14 & $[2{\times}10^4, 10^{-4}, 0.1]$ & $10^{-3} I_4$ & Max $x$ precision. \\
15 & $[10^{-4}, 10^{-4}, 10^{-4}]$ & $10^5 I_4$ & Min input, ignore state. \\
\bottomrule
\end{tabular}
\end{table}

\subsection{Additional Results from Real-World Experiments from Subsection V.\ref{subsec:behavior_real}} \label{app:real}
This section supplements subsection V.\ref{subsec:behavior_real} with data collected from the physical ROSbot XL platform and the corresponding performance of all controllers in real deployment and Gazebo simulations. 

We use only $N=24$ input-state samples collected from the real robot (Figures~\ref{fig:real_data_collection_states} and ~\ref{fig:real_data_collection_inputs}). The same data matrices $U_0$, $X_0$, and $X_1$ were used to synthesize controllers based on Theorem 1, Theorem 2, standard model-based LQR (using identified system and input matrices $\hat{A}$, $\hat{B}$), and Direct Data-Driven LQR (DDLQR)~\cite{ESMZAD2025112197} as a baseline.

For NGD controllers, hyperparameters were chosen to balance control effort and convergence empirically: Theorem 1 used $\alpha=10^{-5}$, $\lambda=0.9$ and Theorem 2 used $\alpha=10^{-6}$, $\lambda=0.6$. For LQR and DDLQR, the weighting matrices were set as $Q = \operatorname{diag}(1,1,10)$ and $R = 0.1 \operatorname{diag}(1,1,1,1)$, with DDLQR~\cite{ESMZAD2025112197} incorporating a discount factor of $0.9$. A common disturbance covariance matrix $W$ was used across all methods $W = \operatorname{diag}(0.001,0.001,0.005)$. 

\begin{figure}[h]
    \centering
    \includegraphics[width=0.9\linewidth]{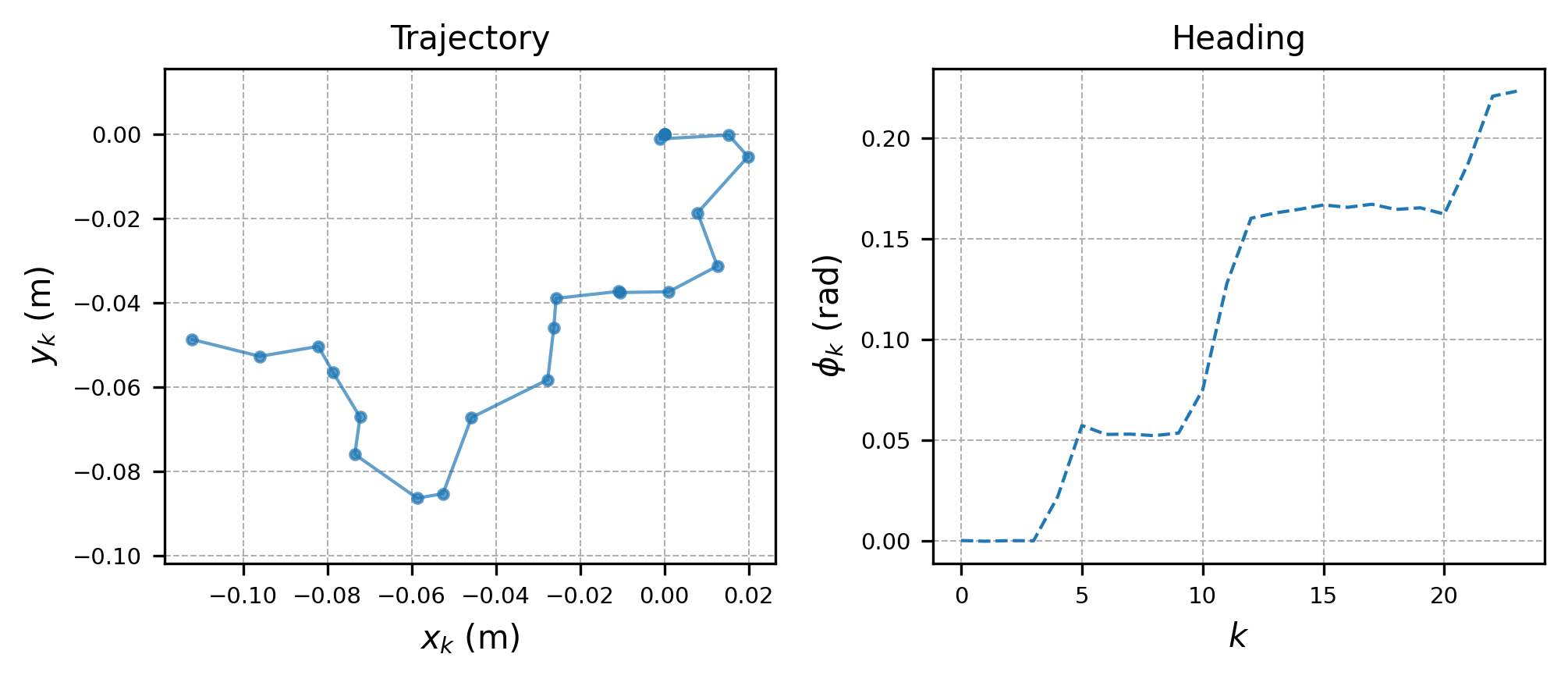}
        \caption{Data collected from the real robot. Left: robot trajectory in the $(x_k, y_k)$ plane during the data collection phase. Right: the corresponding heading angle $\phi_k$ over time.}

    \label{fig:real_data_collection_states}
\end{figure}

\begin{figure}[h]
    \centering
    \includegraphics[width=0.9\linewidth]{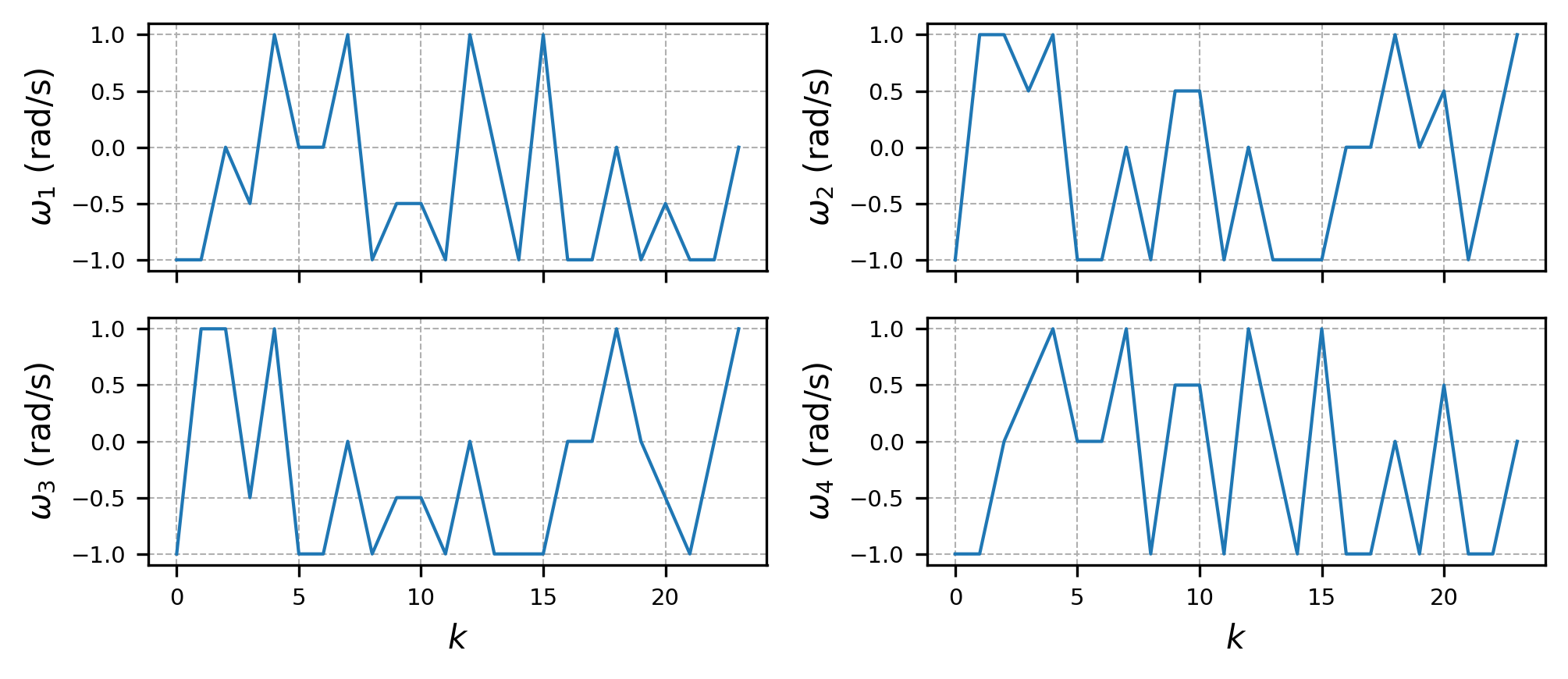}
     \caption{Control inputs (wheel speeds) $\omega_1$ to $\omega_4$ applied to the robot in real-world experiment during data collection. The inputs were selected from structured translation-only (with slight rotational) patterns to preserve the validity of the linear model.}
    \label{fig:real_data_collection_inputs}
\end{figure}

\bibliographystyle{ieeetr} 
\bibliography{refs}

\begin{IEEEbiography}[{\includegraphics[width=1in,height=1.25in,clip,keepaspectratio]{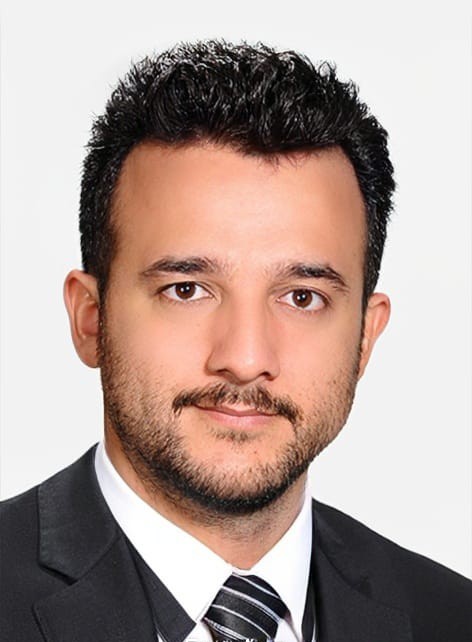}} ]{Ramin Esmzad  } 
     received his Ph.D. in Mechanical Engineering from Michigan State University, USA, and his Ph.D. in Electrical Engineering from Tabriz University of Technology, Iran. He is a Control Engineer at Drive System Design (DSD), where he works on hybrid/electric powertrain optimization and control. His research expertise includes Bayesian filtering and estimation methods, data-driven and safe optimal control, and model predictive control. He also has extensive experience in robotics and cloud technologies, particularly AWS-based platforms for large-scale control and simulation. Ramin serves as a reviewer for leading international journals, including IEEE Transactions on Automatic Control and Automatica. His professional interests span reinforcement learning, real-time optimization, and intelligent control strategies for sustainable mobility and robotic applications. 
\end{IEEEbiography}

\begin{IEEEbiography}[{\includegraphics[width=1in,height=1.25in,clip,keepaspectratio]{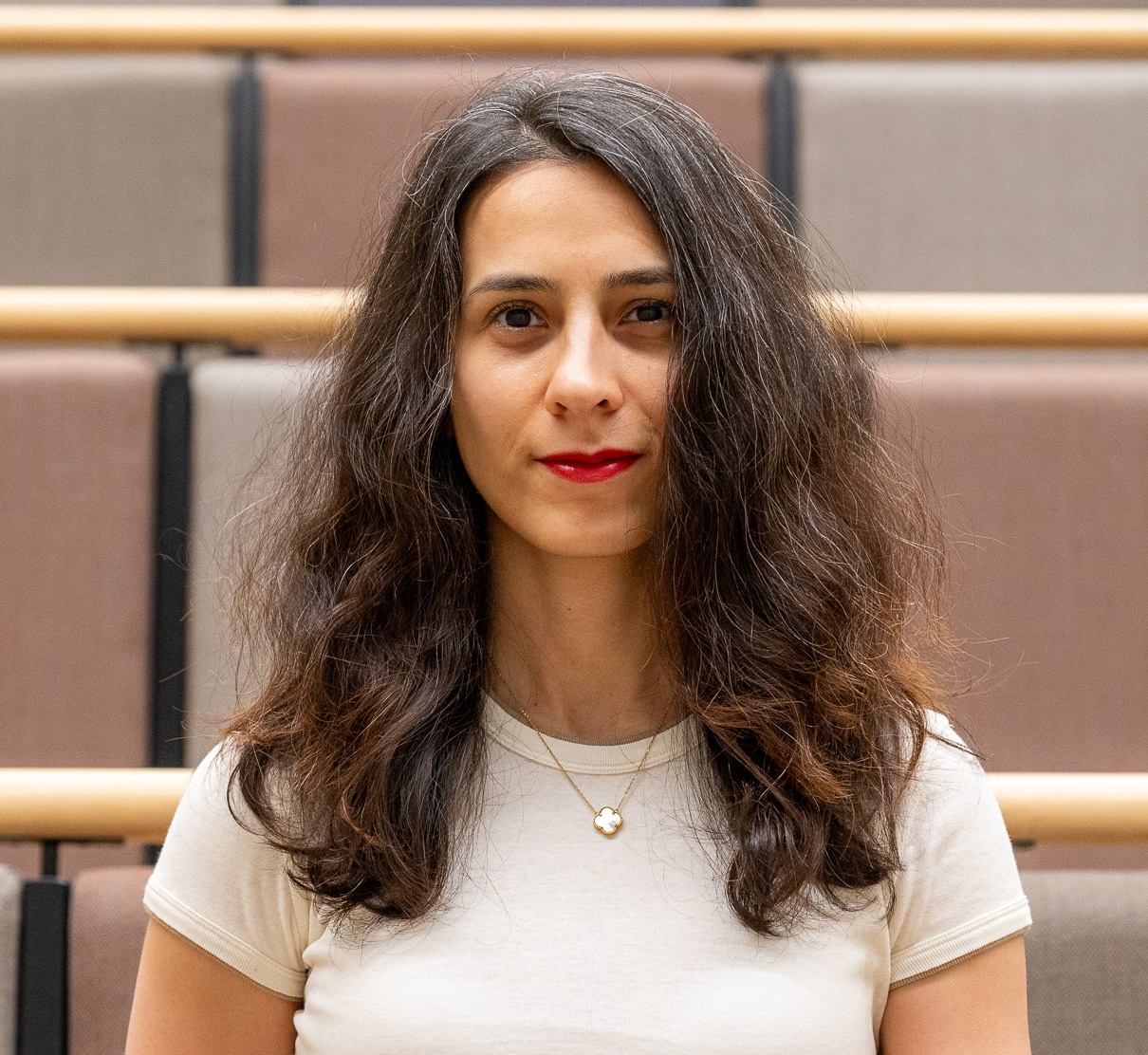}} ]{Farnaz Adib Yaghmaie}{\space} is an Assistant Professor in the Department of Electrical Engineering at Linköping University, Sweden. She earned her Ph.D. in Electrical and Electronic Engineering from Nanyang Technological University (NTU), Singapore, in 2017, where she received the Best Thesis Award among 160 candidates. From 2018 to 2021, she was a postdoctoral researcher in the Division of Automatic Control at Linköping University. Her current research focuses on the intersection of control theory and machine learning, with particular interest in learning methods for control, reinforcement learning, behavioral cloning, and generative models for dynamical systems.
\end{IEEEbiography}

\begin{IEEEbiography}[{\includegraphics[width=0.9in,height=1.15in,clip,keepaspectratio]{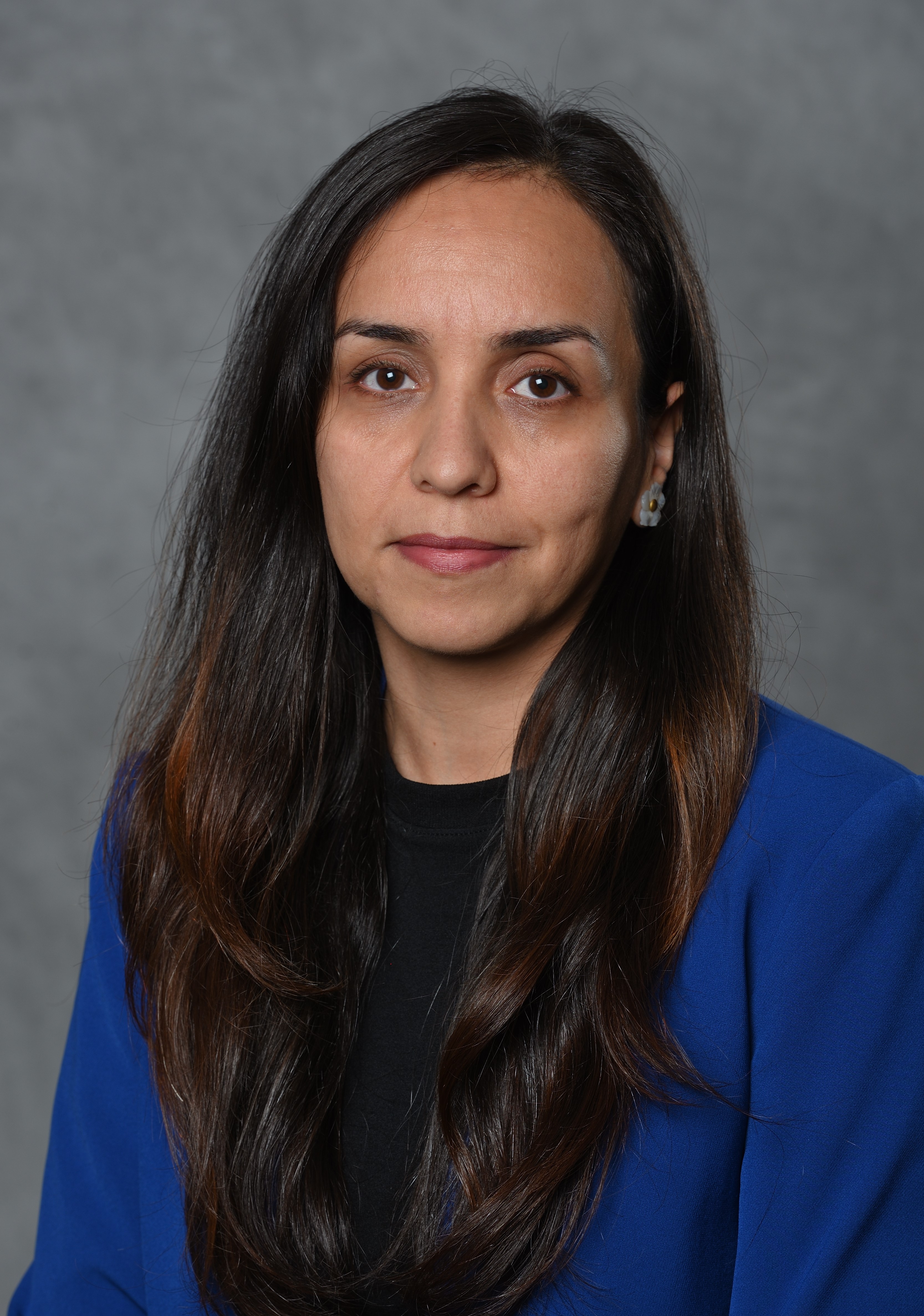}}]{Bahare Kiumarsi } 
 received her B.S. degree in Electrical Engineering from Shahrood University of Technology, Shahrood, Iran, in 2009, her M.S. degree from Ferdowsi University of Mashhad, Iran, in 2013, and her Ph.D. degree from the University of Texas at Arlington, Arlington, TX, USA, in 2017. She is currently an Assistant Professor in the Department of Electrical and Computer Engineering at Michigan State University. Prior to joining Michigan State, she was a Postdoctoral Research Associate at the University of Illinois at Urbana-Champaign. Her research interests include learning-based control, the security of cyber-physical systems, and distributed control of multi-agent systems. She serves as an Associate Editor for Neurocomputing.
\end{IEEEbiography}

\begin{IEEEbiography}[{\includegraphics[width=1in,height=1.25in,clip,keepaspectratio]{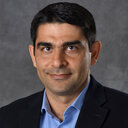}}]{Hamidreza Modares }
 received a B.S. degree from the University of Tehran, Tehran, Iran, in 2004, an M.S. degree from the Shahrood University of Technology, Shahrood, Iran, in 2006, and a Ph.D. degree from the University of Texas at Arlington, Arlington, TX, USA, in 2015, all in Electrical Engineering. He is currently an Associate Professor in the Department of Mechanical Engineering at Michigan State University. Before joining Michigan State University, he was an Assistant professor in the Department of Electrical Engineering at Missouri University of Science and Technology. His current research interests include reinforcement learning, safe control, machine learning in control, distributed control of multi-agent systems, and robotics. During the past five years, he has served as an Associate Editor for IEEE Transactions on Neural Networks and Learning Systems, Neurocomputing, and IEEE Transactions on Systems, Man, and Cybernetics: Systems.
\end{IEEEbiography}

\end{document}